\documentclass[12pt, reqno]{amsart}
\usepackage{amssymb}
\usepackage{graphicx} % to input postscript files
\usepackage{color}
\usepackage{url}
\usepackage{fancyref}
\usepackage{graphicx}
 \usepackage{epstopdf}
\usepackage{dsfont}
\usepackage{amsfonts}
\usepackage{enumerate}
\usepackage[colorlinks = true,
            linkcolor = blue,
            urlcolor  = blue,
            citecolor = blue,
            anchorcolor = blue]{hyperref}
\usepackage{pgfplotstable}
\usepackage{array}
\usepackage{adjustbox}

\numberwithin{equation}{section}

\newtheorem{lemma}{Lemma}
\newtheorem{remark}{Remark}

\newtheorem{theorem}{Theorem}
\newtheorem{corollary}{Corollary}
\newtheorem{proposition}{Proposition}

\begin{document}

\baselineskip 18pt

\newcommand{\E}{\mathbb{E}}
\newcommand{\Eof}[1]{\mathbb{E}\left[ #1 \right]}
\newcommand{\Et}[1]{\mathbb{E}_t\left[ #1 \right]}
\renewcommand{\H}{\mathbb{H}}
\newcommand{\R}{\mathbb{R}}
\newcommand{\sigl}{\sigma_L}
\newcommand{\BS}{\rm BS}
\newcommand{\rv}{{\rm RV}}
\newcommand{\p}{\partial}
\renewcommand{\P}{\mathbb{P}}
\newcommand{\Pof}[1]{\mathbb{P}\left[ #1 \right]}
\newcommand{\var}{{\rm var}}
\newcommand{\cov}{{\rm cov}}
\newcommand{\beaa}{\begin{eqnarray*}}
\newcommand{\eeaa}{\end{eqnarray*}}
\newcommand{\bea}{\begin{eqnarray}}
\newcommand{\eea}{\end{eqnarray}}
\newcommand{\ben}{\begin{enumerate}}
\newcommand{\een}{\end{enumerate}}
\newcommand{\bit}{\begin{itemize}}
\newcommand{\eit}{\end{itemize}}

\newcommand{\bsb}{\boldsymbol{b}}
\newcommand{\bX}{\boldsymbol{X}}
\newcommand{\bx}{\boldsymbol{x}}
\newcommand{\by}{\boldsymbol{y}}
\newcommand{\bE}{\mathbf{e}}
\newcommand{\bw}{\mathbf{w}}
\newcommand{\bsR}{\boldsymbol{R}}
\newcommand{\bW}{\boldsymbol{W}}
\newcommand{\bB}{\boldsymbol{B}}
\newcommand{\bZ}{\boldsymbol{Z}}
\newcommand{\bH}{\mathbf{H}}
\newcommand{\bF}{\mathbf{F}}
\newcommand{\bG}{\mathbf{G}}
\newcommand{\bs}{\mathbf{s}}
\newcommand{\bt}{\mathbf{t}}
\newcommand{\bsihat}{\widehat{\bs_i}}
\newcommand{\bseta}{\boldsymbol{\eta}}

\newcommand{\sgn}{\mbox{sgn}}

\newcommand{\mt}{\mathbf{t}}
\newcommand{\mS}{\mathbb{S}}

\newcommand{\argmax}{{\rm argmax}}
\newcommand{\argmin}{{\rm argmin}}

\newcommand{\tM}{\widetilde{M}}
\newcommand{\tEof}[1]{\tilde{\mathbb{E}}\left[ #1 \right]}
\newcommand{\tP}{\tilde{\mathbb{P}}}
\newcommand{\tW}{\tilde{W}}
\newcommand{\tB}{\tilde{B}}
\newcommand{\1}{\mathbf{1}}
\renewcommand{\O}{\mathcal{O}}
\newcommand{\dt}{\Delta t}
\newcommand{\tr}{{\rm tr}}

\newcommand{\Xv}{X^{(v)}}
\newcommand{\Xvs}{X^{(v^*)}}
\newcommand{\Jv}{J^{(v)}}

\newcommand{\cG}{\mathcal{G}}
\newcommand{\cF}{\mathcal{F}}
\newcommand{\cL}{\mathcal{L}}
\newcommand{\cLv}{\mathcal{L}^{(v)}}

\newcommand{\ttheta}{\tilde{\theta}}
\newcommand{\vega}{{\rm vega}}

\newcommand{\inn}[1]{\langle #1 \rangle}

%
%%***************************************************************************
%%
%%  Title Page
%%
%%***************************************************************************
%

\begin{titlepage}

\begin{center}
\large \bf Target volatility option pricing in lognormal fractional SABR model
\end{center}
\vspace{3mm}

\begin{center}
Elisa Al\`os \\
Departament d'Economia i Empresa and Barcelona Graduate School of Economics \\
Universitat Pompeu Fabra \\
%c/Ramon Trias Fargas, 25-27, 08005 Barcelona, Spain \\
e-mail: \textsf{elisa.alos@upf.edu}
\end{center}

\vskip 2mm

\begin{center}
Rupak Chatterjee \\
Division of Financial Engineering \\
Center for Quantum Science and Engineering\\
Stevens Institute of Technology \\
%Castle Point on Hudson, Hoboken, NJ07030 \\
e-mail: \textsf{Rupak.Chatterjee@stevens.edu}
\end{center}

\vskip2mm

\begin{center}
Sebastian Tudor \\
Division of Financial Engineering \\
Stevens Institute of Technology \\
%Castle Point on Hudson, Hoboken, NJ07030 \\
e-mail: \textsf{studor@stevens.edu}
\end{center}

\vskip2mm

\begin{center}
Tai-Ho Wang \\
Baruch College, CUNY, New York \\
Ritsumeikan University, Shiga, Japan \\
%1 Bernard Baruch Way, New York, NY10010 \\
e-mail: \textsf{tai-ho.wang@baruch.cuny.edu}
\end{center}

\vskip 5mm
%\begin{center}
%Current Version: \today\\
%File Reference: \verb+TVOfSABR.tex+ \centerline{DRAFT: PLEASE
%DO NOT CIRCULATE} 
%\vskip 1cm
%\end{center}
\begin{center}
{\bf Abstract}
\end{center}
We examine in this article the pricing of target volatility options in the lognormal fractional SABR model. 
A decomposition formula by It\^o's calculus yields a theoretical replicating strategy for the target volatility option, assuming the accessibilities of all variance swaps and swaptions. The same formula also suggests an approximation formula for the price of target volatility option in small time by the technique of freezing the coefficient.
Alternatively, we also derive closed formed expressions for a small volatility of volatility expansion of the price of target volatility option. Numerical experiments show accuracy of the approximations in a reasonably wide range of parameters.

\vskip 0.2cm

%\vskip 5mm
%
%\noindent {\it Keywords}: Lognormal fractional SABR model, Decomposition formula, Target volatility option, Small volatility of volatility approximation
%
%\vskip 1cm

%\noindent We are grateful
%for helpful discussions with Marco Avellaneda, Rafael Douady, Andr\'e Lesnewski, and Bruno Dupire
%and the participants at Bloomberg LP and Courant Mathematical Finance seminars.
%We are particularly grateful to Peter Carr for valuable comments and suggestions.
%All errors are our responsibility.

\end{titlepage}

\allowdisplaybreaks

%%***************************************************************************
%%
%%  Document begins here
%%
%%***************************************************************************

%%***************************************************************************

%%\numberwithin{equation}{section}
%%\newtheorem{theorem}{Theorem}[section]

%***************************************************************************

\title[TVO pricing in lognormal fSABR model]{Target volatility option pricing in lognormal fractional SABR model}
\begin{abstract}
We examine in this article the pricing of target volatility options in the lognormal fractional SABR model. 
A decomposition formula by It\^o's calculus yields a theoretical replicating strategy for the target volatility option, assuming the accessibilities of all variance swaps and swaptions. The same formula also suggests an approximation formula for the price of target volatility option in small time by the technique of freezing the coefficient.
Alternatively, we also derive closed formed expressions for a small volatility of volatility expansion of the price of target volatility option. Numerical experiments show accuracy of the approximations in a reasonably wide range of parameters. 
\end{abstract}

%\vskip1mm

\author[E. Al\`os]{Elisa Al\`os}
\author[R. Chatterjee]{Rupak Chatterjee}
\author[S. Tudor]{Sebastian Tudor}
\author[T.-H. Wang]{Tai-Ho Wang}

\address{Elisa Al\`os \newline
Departament d'Economia i Empresa and Barcelona Graduate School of Economics, \newline
Universitat Pompeu Fabra, \newline
c/Ramon Trias Fargas, 25-27, 08005 Barcelona, Spain
}
\email{elisa.alos@upf.edu}

\address{Rupak Chatterjee \newline
Division of Financial Engineering \newline
Stevens Institute of Technology \newline
Castle Point on Hudson, Hoboken, NJ07030
}
\email{Rupak.Chatterjee@stevens.edu}

\address{Sebastian Tudor \newline
Division of Financial Engineering \newline
Stevens Institute of Technology \newline
Castle Point on Hudson, Hoboken, NJ07030
}
\email{studor@stevens.edu}

\address{Tai-Ho Wang \newline
Department of Mathematics \newline
Baruch College, The City University of New York \newline
1 Bernard Baruch Way, New York, NY10010 \newline
and \newline
Department of Mathematical Sciences \newline
Ritsumeikan University \newline
Noji-higashi 1-1-1, Kusatsu, Shiga 525-8577, Japan
}
\email{tai-ho.wang@baruch.cuny.edu}


\maketitle

%\begin{center}
%\today
%\end{center}

%%%%%%%%%%%%%%%%%%%%%%%%%%%%%%%%%%%%%%%%%%%%%%%%%%%%%%%%%%%%%%%%%
%
%  Main text starts here
%
%%%%%%%%%%%%%%%%%%%%%%%%%%%%%%%%%%%%%%%%%%%%%%%%%%%%%%%%%%%%%%%%%

%%%%%%%%%%%%%%%%%%%%%%%%%%%%%%%%%%%%%%%%%%%%%%%%%%%%%%%%%%%%%%%%%
%  Section: Introduction
%%%%%%%%%%%%%%%%%%%%%%%%%%%%%%%%%%%%%%%%%%%%%%%%%%%%%%%%%%%%%%%%%

\section{Introduction}
Target volatility options (TVOs) are a type of derivative instrument that explicitly depends on the evolution of an underlying asset as well as its realized volatility. This option allows one to set a ‘target’ volatility parameter that determines the leverage of an otherwise price dependent payoff. The multiplicative leverage factor is the ratio of the target volatility to the realized volatility of the underlying asset at the maturity of the option.  If this target volatility is chosen to be the implied market volatility of the underlying asset, this option becomes similar to pure volatility instruments such as variance and volatility swaps where investors are swapping realized volatility for implied volatility. TVOs are slightly different as they do not explicitly perform a swap but rather consider the ratio of the two types of volatilities in order to increase or decrease a potential price payoff.   

The typical form of the TVO leverage factor has the target volatility parameter in the numerator and the realized volatility in the denominator. When an asset exhibits a smooth general trend upwards, its realized volatility tends to decrease thereby giving the European call version of a TVO a greater payoff (levered) than the simple vanilla version of a European call. When viewing the target volatility parameter from an implied volatility point of view, the TVO leverage factor allows an investor to possibly recoup the expensive premium of a call option that was bought during a high volatility regime. For put style TVOs, the typical form of the leverage factor tends to work in the opposite way (i.e. de-levering). Assets tend to move more erratically downwards thereby increasing their realized volatilities in a bear market and thereby decreasing the put payoff. One could therefore imagine creating put style TVOs where the volatility leverage factor is the inverse of the typical form such that the realized volatility is placed in the numerator and the target volatility is in the denominator. This is a substantially different structure than the typical call option TVO and therefore is a topic of future research and will not be discussed here.   

The risk management of TVOs is difficult as one cannot simply perform a standard delta-gamma-vega style hedge as this does not fully take into account the risk embedded in the volatility leverage factor. In this paper, we propose a static volatility hedge of TVOs using an identical maturity variance swap. This hedge becomes more accurate when the target volatility parameter is chosen to be close to the square root of the strike of the variance swap. 

The standard stylized facts of equity returns such as volatility clustering, fat tails, and the leverage effect all have important roles to play in the pricing of TVOs. In particular, the temporal correlation of squared returns has an important effect on realized volatilities. Therefore, we have chosen to model the instantaneous volatility process using fractional Brownian motion (fBM) as it has a very simple and appealing correlation structure in terms of Hurst exponents. Ideally, one would like to produce implied Hurst parameters from market prices in order to quantify the temporal correlation dependence of TVOs. As the TVO market is still quite exotic, this will have to wait until this OTC market becomes somewhat more liquid.    
Finally, the explicit option payoff dependence on volatility and strike requires one to use a stochastic volatility (driven by fBM) correlated to a stochastic price diffusion model. In this paper, we investigate the fractional SABR (fSABR) model that was recently suggested and empirically test against market data in \cite{gjr}. We refer the reader to \cite{asw} for more detailed discussions on the probability density of the lognormal fSABR model.

In \cite{dgt}, the authors provide prices for TVOs using a single-factor stochastic volatility model where the instantaneous volatility of the underlying asset is assumed to be, unrealistically, independent from the Brownian motion driving the asset returns. This unrealistic assumption for TVOs is addressed in \cite{t} using such models as the Heston model and the 3/2 stochastic volatility model. In \cite{gr}, this approach is further enhanced with the addition of stochastic skew via a multi-factor stochastic volatility model.
The authors of \cite{fgg} price various products characterized by payoffs depending on both a stock and its volatility, TVO being one such case, using a Fourier-analysis approach. A variance-optimal hedge approach for TVOs under exponential Levy dynamics is investigated in \cite{ww}.

To our knowledge, none of the existing literatures in TVO pricing deals with fractional volatility process. With the introduction of fractional process to the instantaneous volatility, it comes with Hurst exponent risk. In other words, how much does it affect the TVO price if the Hurst exponent is misspecified? Is there a replicating/hedging strategy that is relatively robust against the Hurst exponent? We do not intend to answer all these questions in the current paper. However, the decomposition formula given by \eqref{eqn:decomp-pathwise} in Section \ref{sec:decomp-formula} suggests theoretically a robust replicating strategy for TV call given the accessibility of all swaps and swaptions. 
%{\color{red} Up to my knowledge, the only paper on hedging with fractional noises is the one by El Euch- Rosenbaum: https://arxiv.org/abs/1703.05049, but the hedging strategy is not so easy. I know there are some researchers working on this problem, but with still with no available results}

The rest of the paper is organized as follows. Section \ref{sec:model-spec} briefly specifies the model and introduces pricing of TVO. Section \ref{sec:decomp-formula} shows the decomposition formula for the price of a TV call in terms of It\^o calculus and of Malliavin derivative. By specifying the model to the lognormal fSABR model, we derive approximations of the price of a TV call in Sections \ref{sec:tvo-fsabr} and \ref{sec:small-vol-vol}. The paper concludes with numerical implementations and discussions in Sections \ref{sec:numerics} and \ref{sec:discussion}. Technical calculations and proofs are left as an appendix in Section \ref{sec:appendix}.

%%%%%%%%%%%%%%%%%%%%%%%%%%%%%%%%%%%%%%%%%%%%%%%%%%%%%%%%%%%%%%%%%
%  Section: Model specification and TV options
%%%%%%%%%%%%%%%%%%%%%%%%%%%%%%%%%%%%%%%%%%%%%%%%%%%%%%%%%%%%%%%%%

\section{Model specification and target volatility options} \label{sec:model-spec}
Throughout the text, $B_t$ and $W_t$ denote two independent standard Brownian motions defined on the filtered probability space $(\Omega, \cF_t, \mathbb Q)$ satisfying the usual conditions. All random variables and stochastic processes are defined over $(\Omega, \cF_t, \mathbb Q)$. Denote by $S_t$ the price process of the underlying asset and $\alpha_t$ the instantaneous volatility. Risk free rate is assumed zero for simplicity, thus the evolution of $S_t$ under the risk neutral probability $\mathbb Q$ is governed by 
\[
\frac{dS_t}{S_t} = \alpha_t d\tilde W_t = \alpha_t(\rho dB_t + \bar\rho dW_t),
\]
where $\rho\in (-1,1)$ and $\bar\rho = \sqrt{1 - \rho^2}$. At the moment, other than positivity and technical conditions, we do not specify the dynamic for $\alpha_t$ just yet but assuming it is a square-integrable process adapted to the filtration generated by the Brownian motion $B$. For computational purpose, further specification of $\alpha_t$ is considered in Section \ref{sec:tvo-fsabr} and thereafter.

For fixed $K > 0$, define $X_t := \log\frac{S_t}K$ and $Y_t := \alpha_t$. Then $X_t$ satisfies
\begin{equation}
dX_t = Y_t d\tilde W_t - \frac{Y_t^2}2 dt. \label{eqn:model-x}
\end{equation}
We shall be mostly working with the $X$, $Y$ variables in the following. A {\it target volatility} (TV) call struck at $K$  pays off at expiry $T$ the amount
\begin{equation} \label{eqn:tvo-payoff}
\frac{\bar\sigma}{\sqrt{\frac1T\int_0^T \alpha_t^2 dt}} \left( S_T - K \right)^+ = \frac{K \, \bar\sigma \sqrt T}{\sqrt{\int_0^T Y_t^2 dt}} \left( e^{X_T}- 1 \right)^+,
\end{equation}
where $\bar\sigma$ is the (preassigned) {\it target volatility} level. Apparently, if at expiry the realized volatility is higher (lower) than the target volatility, the payoff is scale down (up) by the ratio between target volatility and realized volatility. For $t \leq T$, the price at time $t$ of a TV call struck at $K$ with expiry $T$ is hence given by the conditional expectation under risk neutral probability $\mathbb Q$ as 
\begin{equation}
K\, \bar\sigma \sqrt T \, \Eof{\left. \frac1{\sqrt{\int_0^T Y_\tau^2 d\tau}} \left( e^{X_T}- 1 \right)^+\right|\cF_t} \label{eqn:tvo-price-t}
\end{equation} 
provided the expectation is finite.
%On the other hand, the payoff function and the price of a TV put are given respectively by  \\
%\blue{(Q: Are we still arguing that the scaling factor in the payoff of a TV put should be the other round?)}

By temporarily ignoring the constant factor out front, we shall evaluate the conditional expectation in \eqref{eqn:tvo-price-t} in the following sections. 

%%%%%%%%%%%%%%%%%%%%%%%%%%%%%%%%%%%%%%%%%%%%%%%%%%%%%%%%%%%%%%%%%
%  Section: The decomposition formula 
%%%%%%%%%%%%%%%%%%%%%%%%%%%%%%%%%%%%%%%%%%%%%%%%%%%%%%%%%%%%%%%%%

\section{The decomposition formula} \label{sec:decomp-formula}
In the spirit of \cite{alos}, we derive in this section a decomposition formula for TV calls which in turn leads to a theoretical replicating strategy assuming the accessibilities of all variance swaps and swaptions. An approximation of the price of a TV call is obtained in Theorem \ref{thm:TVcall-approx-decomp} by ``freezing the coefficients".  

The following notations will also be used throughout the rest of the paper. 
The normalized Black-Scholes function $C = C(x,w)$ is defined as 
\begin{equation}
C(x,w) := e^x N(d_1) - N(d_2), \label{eqn:BS-func}
\end{equation}
where $d_1 = \frac x{\sqrt w} + \frac{\sqrt w}2$ and $d_2 = d_1 - \sqrt w$. Note that $C$ satisfies the (forward) Black-Scholes PDE
\begin{equation}
C_w - \frac12 C_{xx} + \frac12 C_x = 0 \label{eqn:BS-eqn}
\end{equation}
with initial condition $C(x,0) = (e^x - 1)^+$.  
For any $t \in [0,T]$, define 
\beaa 
&& w_t := \int_0^t Y_s^2 ds, \quad 
\hat{w}_t := \int_t^T \Et{Y_s^2} ds, \\ 
&& M_t := \int_0^T \Et{Y_s^2} ds = w_t + \hat w_t,
\eeaa
where $\Et{\cdot}$ is a shorthand notation for the conditional expectation $\Eof{\cdot|\cF_t}$, for $0 \leq t \leq T$.
We remark that $M$ is a martingale and note that $w_t$ represents the integrated/realized variance up to time $t$, whereas $\hat w_t$ represents the price of the variance swap (zero interest rate) observing the time period $[t,T]$. Also notice that, at $t = 0$, $M_0 = \hat w_0$ which equals the price of the variance swap between $t=0$ and $t=T$. 
Furthermore, we define $F = F(x,w,\hat w)$ as 
\begin{equation}
F(x,w,\hat w) := \frac{C(x,\hat w)}{\sqrt{w + \hat w}}. \label{eqn:func-F}
\end{equation}
Notice that since $\hat w_t = M_t - w_t$, we have $d\hat w_t = dM_t - dw_t$. Thus, $d\inn{X, \hat w}_t = d\inn{X, M}_t$ and $d\inn{\hat w}_t = d\inn{M}_t$ since $w$ is of finite variation.  

We will impose the following hypotheses.
\begin{itemize}
\item[({\bf H1})] Assume that the process $Y$ has a martingale representation of the form
$$
Y_t = \Eof{Y_t} + \nu\int_0^t a(t,s)dB_s,
$$
where $\nu>0$ and, for any $t$,  $a(t,\cdot )$ is an adapted process satisfying 
\begin{equation}
\label{H}
|a(t,s)|\leq A|t-s|^\delta
\end{equation}
for some $\delta\in (-\frac12,\frac12)$ and for some random variable $A\in \cap_{p\geq 1}L^p$.

\item[({\bf H2})] $\frac{1}{\hat w_t}+\frac{1}{w_T-w_t}\leq  \frac{A}{T-t}$, for some random variable $A\in \cap_{p\geq 1}L^p$.
\end{itemize}

%\begin{remark}
%Notice that {\rm (\bf H1)} holds if, for example, $Y=f(\nu B^H)$, where $f$ is a bounded function and $B^H$ is a fractional Brownian motion with $H=\delta+\frac{1}{2}$ (see Section \ref{sec:tvo-fsabr}).
%\end{remark}

\begin{remark}
Notice that (H1) implies that
$$
dM_t=\nu\left(\int_t^T a(u,t)du\right)dB_t.
$$
\end{remark}

Theorem \ref{thm:TV-call-decomp} below presents a decomposition formula for TV call in terms of the functions $C$ and $F$.
\begin{theorem}(Decomposition formula for TV call) \label{thm:TV-call-decomp} \\
Consider the model (\ref{eqn:model-x}) and assume that (H1) and (H2) hold. Then the price of the target volatility call at time $t$ struck at $K$ with expiry $T$ can be decomposed as  
\bea
&& K \, \bar\sigma \sqrt{T} \Et{\frac1{\sqrt{w_T}} \left( e^{X_T} - 1 \right)^+} \nonumber \\
&=& K\bar\sigma \sqrt{T}
\left[\frac1{\sqrt{M_t}} C(X_t, \hat{w}_t) + \Et{\int_t^T F_{x\hat w}d\inn{X,M}_s + \frac12 \int_t^T F_{\hat w \hat w}d\inn{M}_s}\right]\label{decomposition}
\eea
which apparently extends the deterministic volatility case. Note that the decomposition formula \eqref{decomposition} does not depend on the specification of volatility process as long as the corresponding integrals are defined. Thus the first term in \eqref{decomposition} represents the price of a TV call using variance swap as future (deterministic) realized variance.
\end{theorem}
\begin{proof} This proof is similar to the proof of Theorem 3.1 in \cite{alos}, so we only sketch it. Denote 
$ w^\epsilon_t:=\epsilon+w_t$, $\hat w^\epsilon_t:=\epsilon+\hat w_t$ and $M_t^\epsilon:= w^\epsilon_t+ \hat w^\epsilon_t$. 
By applying It\^o's formula to the process $F(X_t,  w^\epsilon_t, \hat  w^\epsilon_t)$ we obtain
\beaa
&& \frac1{\sqrt{ w^\epsilon_T}}\left( e^{X_T} - 1 \right)^+ - F(X_t,  w^\epsilon_t, \hat w^\epsilon_t) \\
&=& \int_t^T F_x dX_s + \int_t^T F_w dw_s + \int_t^T F_{\hat w} d\hat w_s \\
&& + \frac12 \int_t^T F_{xx} d\inn{X}_s + \int_t^T F_{x\hat w}d\inn{X,M}_s + \frac12 \int_t^T F_{\hat w \hat w}d\inn{M}_s \\
&=& \int_t^T \frac{C_x}{\sqrt{M^\epsilon_s}} \left( Y_s d\tilde W_s - \frac{Y_s^2}2 ds \right) - \frac12 \int_t^T\frac{C}{M_s^{3/2}} Y_s^2 ds \\
&& + \int_t^T\left(\frac{C_w}{\sqrt{M_s^\epsilon}} - \frac12\frac{C}{(M_s\epsilon)^{3/2}}\right) \left( dM_s - Y_s^2 ds \right) \\
&& + \frac12 \int_t^T \frac{C_{xx}}{\sqrt{M_s^\epsilon}} Y_s^2 ds + \int_t^T F_{x\hat w}d\inn{X,M}_s + \frac12 \int_t^T F_{\hat w \hat w}d\inn{M}_s,
\eeaa
where $d\tilde W_t = \rho dB_t + \bar\rho dW_t$. The function $C$ and all its partial derivatives are evaluated at $(X_t,\hat w_s)$ whereas $F$ and all its partial derivatives are evaluated at $(X_s,  w^\epsilon_s, \hat  w^\epsilon_s)$. Note that since $C$ satisfies the Black-Scholes equation \eqref{eqn:BS-eqn}, it follows that
\begin{eqnarray}
&& \frac1{\sqrt{ w^\epsilon_T}}\left( e^{X_T} - 1 \right)^+ - F(X_t, w^\epsilon_t, \hat  w^\epsilon_t) \label{eqn:decomp-pathwise} \\
&=& \int_t^T F_x Y_s d\tilde W_s + \int_t^T F_{x\hat w} dM^\epsilon_s + \int_t^T F_{x\hat w}d\inn{X,M^\epsilon}_s + \frac12 \int_t^T F_{\hat w \hat w}d\inn{M^\epsilon}_s. \nonumber 
\end{eqnarray}
Finally, by taking conditional expectations on both sides of \eqref{eqn:decomp-pathwise}, letting $\epsilon\to 0$, and using the fact that 
$$
|\partial^n_x \partial_w C (X_t,w_t)|\leq C w_t^{-\frac{1}{2} (n+1)},
$$
for some positive constant $C$,
we obtain the decomposition formula \eqref{decomposition}.
\end{proof}\begin{remark}
The formula \eqref{eqn:decomp-pathwise} suggests a {\it semiparametric} dynamical replicating strategy so long as all variance swaps and swaptions are accessible. The first term on the right hand of \eqref{eqn:decomp-pathwise} is in fact the classical delta hedge of the underlying in Black-Scholes model but scaled by $\frac1{\sqrt{M_t}}$, whereas the second term can be regarded as a ``delta hedge" of the variance swap. The third is associated with the covariation between the underlying and the volatility process which can be replicated by holding gama swaps. Finally, the last term corresponds to holding time varying shares of log contracts on variance swap.  
\end{remark} 

In terms of Malliavin derivative, Theorem \ref{thm:decomp-malliavin} in the following shows another decomposition formula for the price of a TV call which, in the uncorrelated case $\rho = 0$, coincides with \eqref{uncorrelated}.
Thus this newly derived decomposition can be regarded as the TVO version of the extended Hull and White formula as the case for vanilla options in \cite{alos}. 
We assume that the reader is familiar with the elementary results in Malliavin calculus as given for instance in \cite{nualart}. In the remaining of this paper $\mathbb{D}_{B}^{1,2}$ denotes the domain of the Malliavin derivative operator $D^{B}$ with respect to the Brownian motion $B$. It is well-known that $\mathbb{D}_{B}^{1,2}$ is a dense subset of $L^{2}(\Omega)$ and that $D^B$ is a closed and unbounded operator from $\mathbb{D}_{B}^{1,2}$ to $L^{2}([0,T]\times \Omega )$. We also consider the iterated derivatives $D^{n,B}$, for $n>1,$ whose domains will be denoted ba 
$\mathbb{D}_{B}^{n,2}$. We will use the notation $\mathbb{L}_{B}^{n,2}=$\ $%
L^{2}(\left[ 0,T\right];\mathbb{D}_{B}^{n,2})$.

We will make use of the following anticipating It\^o's formula, see for example \cite{alos06}. We denote $a_t:=w_T-w_t$.
\begin{proposition} \label{Ito}
Assume that in model (\ref{eqn:model-x}) the process $Y$ satisfies $Y^{2}\in \mathbb{L}^{1,2}_B$. Let
$g=g(t,x,a):[0,T]\times \mathbb{R}^{2} \rightarrow \mathbb{R}$ 
be a function
in $C^{1,2} ([0,T]\times \mathbb{R}^{2})$ statisfying that there exists a
positive constant $M$ such that, for all $t\in \left[ 0,T\right]
,$ $g$ and its partial derivatives evaluated in $\left(
t,X_{t},a_{t}\right)$ are bounded by $M$. It follows that
\begin{eqnarray}
g(T,X_{T},a_T) &=& g(t,X_{t},a_{t}) + \int_{t}^{T} g_t(s,X_{s},a_{s})ds \nonumber\\
&& -\int_{t}^{T}g_x(s,X_{s},a_{s}) \frac{Y_{s}^{2}}{2}ds %\nonumber\\ 
+ \int_{t}^{T}g_x(s,X_{s},a_{s})Y _{s} d\tilde W_s \nonumber\\
&& -\int_{t}^{T}g_a(s,X_{s},a_{s})Y_{s}^{2}ds + \rho \int_{t}^{T}g_{xa}(s,X_{s},a_{s})\Theta _{s}ds \nonumber\\
&& +\frac{1}{2}\int_t^T g_{xx}(s,X_{s},a_{s})Y_{s}^{2}ds,
\label{aito}
\end{eqnarray}
where $\Theta _{s}:=(\int_{s}^{T}D^B_{s}Y _{r}^{2}dr)Y _{s}.$ 
\end{proposition}
The decomposition formula for the price of a TV call in terms of Malliavin derivative is then almost a straightforward application of Proposition \ref{Ito}. To that end, the following hypotheses are imposed.
\begin{itemize}
\item[(\bf H1')] Assume that $Y\in\mathbb{L}^{1,2}_B$ and that
$$
|D_s Y_t|\leq A|t-s|^\delta,
$$
for some $\delta\in (-\frac12,\frac12)$ and for some random variable $A\in \cap_{p\geq 1}L^p$.

\item[(\bf H2')] $\frac{1}{w_T-w_t}\leq  \frac{A}{T-t}$, for some random variable $A\in \cap_{p\geq 1}L^p$.

\end{itemize}

\begin{theorem}(Decomposition formula in terms of Malliavin derivative) \label{thm:decomp-malliavin} \\
Consider the model (\ref{eqn:model-x}) and assume that {\rm ({\bf H1'})} and {\rm ({\bf H2'})} hold. 
Then the price of the target volatility call at time $t$ struck at $K$ with expiry $T$ can be decomposed in terms of Malliavin derivative as  
\begin{eqnarray}
\label{decompositionMalliavin}
&& K\bar\sigma \sqrt{T} \, \Et{\frac1{\sqrt{w_T}} \left( e^{X_T} - 1 \right)^+} \nonumber \\
&=& K\bar\sigma \sqrt{T} \, \left[\Et{\frac1{\sqrt{w_T}}C(X_t, a_t)} + \rho \, \Et{\int_t^T F_{x \hat w}(X_s, w_s, a_s) Y_s\left(\int_s^TD^B_s Y_r^2 dr\right)ds}\right],\nonumber
\end{eqnarray}
where the function $F$ is defined in \eqref{eqn:func-F}. Apparently, in the uncorrelated case the second term on the right side of \eqref{decompositionMalliavin} vanishes thus coincides with \eqref{uncorrelated} when $t = 0$. 
\end{theorem}
\begin{proof}
The proof of this theorem is similar to the proof of Theorem 4.2 in \cite{alv}, so we only sketch it. Applying the anticipating It\^{o}'s formula \eqref{aito} with $g$ defined as $g(t, X_t, a_t):= F(X_t, w_t, a_t) = \frac{1}{\sqrt{w_t+a_t}}C(X_t,a_t)$, using the fact that $C$ satisfies the Black Scholes equation \eqref{eqn:BS-eqn} and taking expectations we obtain
\begin{eqnarray*}
\Et{\frac1{\sqrt{w_T}}C(X_T,a_T)} &=& \Et{\frac{C(X_t,a_t)}{\sqrt{w_t + a_t}}} %\nonumber\\
%&+& \Et{\int_t^T \left(-\frac{\sqrt T}{2(\sqrt{w_T})^3} C_x+ \frac{\sqrt T}{\sqrt{w_T}}C_{xw}\right)\Theta_s ds}.
+ \rho \, \Et{\int_t^T F_{x\hat w}(X_s, w_s, a_s) \Theta_s ds}.
\end{eqnarray*}
Now using the fact that $C(X_T,a_T) = \left( e^{X_T}- {1} \right)^+$ since $a_T = 0$, the result follows.
\end{proof}

We conclude the section by an approximation of the price of TV call suggested by the decomposition formula \eqref{decomposition} as $t$ approaches $T$. We will need the following hypothesis
\begin{itemize}
\item[({\bf H3})] Assume ({\bf H1}) holds and, for any $t>s$,
$$
a(t,s)=\E_s (a(t,s))+\nu\int_t^s b(t,s,\theta)dB_\theta,
$$
where $\nu$ is defined as in ({\bf H1}) and for some adapted process $b(t,s,\cdot)$ such that
$$
|b(t,s,\cdot)|\leq A (t-s)^\delta (s-\theta)^\delta,
$$
for some random variable $A\in\cap_{p\geq 1} L^p$.
\end{itemize}
%\begin{remark}
%We notice that {\rm ({\bf H3})} holds if, for example, $Y=f(\nu B^H)$, where $f$ is a bounded function and $B^H$ is a fractional Brownian motion with $H=\delta+\frac{1}{2}$.
%\end{remark}
 
\begin{theorem} \label{thm:TVcall-approx-decomp}
The price of a TV call at time $t$, for $t < T$, struck at $K$ with expiry $T$ has the following approximation as $t$ approaches $T$.
\bea
&& K\bar\sigma \sqrt{T} \, \Et{\frac1{\sqrt{w_T}} \left( e^{X_T} - 1 \right)^+} \label{eqn:approx-tvc} =\Lambda (T,t)+R(T,t),
\eea
where
%\bea
%\Lambda (T,t)&:=& K\bar\sigma \sqrt{T} \, \left[\frac1{\sqrt{M_t}} C(X_t, \hat w_t) + F_{x\hat w}(X_t, w_t, \hat w_t) \, \Et{\inn{X, M}_T - \inn{X, M}_t} \right] \nonumber 
%\eea
\beaa
\Lambda(T,t) &:=& K\bar\sigma \sqrt{T} \left\{ \frac1{\sqrt{M_t}} C(X_t, \hat w_t) + F_{x\hat w}(X_t, w_t, \hat w_t) \, \Et{\inn{X, M}_T - \inn{X, M}_t} \right. \\
&& \left. + \frac12 F_{\hat w\hat w}(X_t, w_t, \hat w_t) \, \Et{\inn{M}_T - \inn{M}_t} \right\} 
\eeaa
and 
$$
 R(T,t)=O(\nu^2+\rho\nu)^2.
$$
In particular, at $t = 0$ the formula slightly simplifies as 
\bea
&&\Eof{\frac1{\sqrt{w_T}} \left( e^{X_T} - 1 \right)^+}  \label{eqn:approx-tvc-tzero} \\
&\approx& \frac1{\sqrt{M_0}} C(X_0, M_0) + F_{x\hat w}(X_0, 0, M_0) \, \Eof{\inn{X, M}_T} 
+ \frac12 F_{\hat w\hat w}(X_0, 0, M_0) \, \Eof{\inn{M}_T} \nonumber 
\eea
since $M_0 = \hat w_0$ and $w_0 = 0$.
\end{theorem}
\begin{proof} 
Hypotheses ({\bf H3}) gives us that 
\begin{eqnarray}
d\inn{X,M}_s&=&\nu Y_s\left(\int_s^T \left(E_t(a(u,s)+ \nu\int_t^s  b(u,s,\theta)dB_\theta\right)du\right)ds\\
&&=:Y_s\Gamma_s ds\nonumber
\end{eqnarray}
and 
\begin{eqnarray}
d\inn{M,M}_s&=&\nu^2\left(\int_s^T \left(E_t(a(u,s)+ \nu\int_t^s  b(u,s,\theta)dB_\theta\right)du\right)^2 ds\\
&&=\Gamma_s^2 ds\nonumber
\end{eqnarray}
It\^{o}'s formula gives us that
\beaa
&&\E\left[ \left(F_{x\hat w}(X_s, w^\epsilon_s, \hat{w}^\epsilon_s) - F_{x\hat w}(X_t, w^\epsilon_t, \hat{w}^\epsilon_t)\right)Y_s\Gamma_s\right]   \\
&&=\E\left[\int_t^s F_{xx\hat w\hat w}Y_s\Gamma_sd\inn{X,M}_s + \frac12 \int_t^s F_{x \hat w\hat w\hat w}Y_s\Gamma_s d\inn{M}_s\right.\\
&&+\int_t^s \left(F_{xx\hat w}\Gamma_ud\inn{X,Y}_u+F_{xx\hat w}Y_ud\inn{X,\Gamma}_u\right)\\
&&+\int_t^s \left(F_{x\hat w\hat w}\Gamma_ud\inn{X,Y}_u+F_{x\hat w\hat w}Y_ud\inn{X,\Gamma}_u\right)\\
&&+\int_t^s \left(F_{x\hat w}(X_u, w^\epsilon_u, \hat{w}^\epsilon_u) - F_{x\hat w^\epsilon}(X_t, w^\epsilon_t, \hat{w}^\epsilon_t)\right) d\inn{Y,\Gamma}_u,
\eeaa
from where it follows that (letting $\epsilon\to 0$)
$$
\int_t^T \E\left[ \left(F_{x\hat w}(X_s, w_s, \hat{w}_s) - F_{x\hat w}(X_t, w_t, \hat{w}_t)\right)Y_s\Gamma_s\right] ds=O(\rho\nu+\nu^2)^2.
$$
In a similar way we can prove that 
$$
\int_t^T \E\left[ \left(F_{\hat w\hat w}(X_s, w_s, \hat{w}_s) - F_{\hat w\hat w}(X_t, w_t, \hat{w}_t)\right)\Gamma^2_s\right] ds=O(\rho\nu+\nu^2)^2,
$$
and now the proof is complete.

\end{proof}
We remark that by straightforward calculations the functions $F_{x\hat w}$ and $F_{\hat w\hat w}$ can be expressed in terms of the Black-Scholes function $C$ as 
\bea
&& F_{x\hat w}(x, w, \hat w) = -\frac{C_x(x,\hat w)}{2(\sqrt{w + \hat w})^3} + \frac{C_{xw}(x,\hat w)}{\sqrt{w + \hat w}}, \label{eqn:Fxw-in-C} \\
&& F_{\hat w\hat w}(x, w, \hat w) = -\frac{C_w(x,\hat w)}{(\sqrt{w + \hat w})^3} + \frac{3C(x,\hat w)}{4(\sqrt{w + \hat w})^5} + \frac{C_{ww}(x,\hat w)}{\sqrt{w + \hat w}}. \label{eqn:Fww-in-C}
\eea
Thus, to numerically implement the approximation formulas \eqref{eqn:approx-tvc} and \eqref{eqn:approx-tvc-tzero}, we will have to be able to compute $M_t$ and the quadratic variations $\inn{X,M}$ and $\inn{M}$. To this end, an explicit specification of the volatility process $Y_t$ (or $\alpha_t$) is required. We specify ourselves to the lognormal fractional SABR model in the following sections.

%%%%%%%%%%%%%%%%%%%%%%%%%%%%%%%%%%%%%%%%%%%%%%%%%%%%%%%%%%%%%%%%%
%  Section: Target volatility option in lognormal fSABR model
%%%%%%%%%%%%%%%%%%%%%%%%%%%%%%%%%%%%%%%%%%%%%%%%%%%%%%%%%%%%%%%%%
\section{Target volatility option in lognormal fSABR model} \label{sec:tvo-fsabr}
In this section, we calculate in details and explicitly the approximation formula \eqref{decomposition} for the price of a TV call under the lognormal fSABR model suggested in \cite{asw} and \cite{gjr}. 
The price of underlying asset $S_t$ and its instantaneous volatility $\alpha_t$ in the lognormal fSABR model are governed by
\beaa
&& \frac{dS_t}{S_t} = \alpha_t (\rho dB_t + \bar\rho dW_t), \\
&& \alpha_t = \alpha_0 e^{\nu B_t^H},
\eeaa
where $B_t$ and $W_t$ are independent Brownian motions, $\bar\rho = \sqrt{1 - \rho^2}$, and $B_t^H$ is the fractional Brownian motion driven by $B_t$. That is, 
\[
B_t^H = \int_0^t K(t,s) dB_s,
\]
where $K$ is the Molchan-Golosov kernel
\begin{equation}
K(t,s) = c_H (t-s)^{H-\frac{1}{2}}F\left(H-\frac{1}{2},\frac{1}{2}-H,H+\frac{1}{2};1-\frac{t}{s}\right)\mathbf{1}_{[0,t]}(s), \label{eqn:molchon-golosov}
\end{equation}
with $c_H=\left[\frac{2H\Gamma\left(\frac{3}{2}-H\right)}{\Gamma(2-2H)\Gamma\left(H+\frac12\right)}\right]^{1/2}$ and $F$ is the Gauss hypergeometric function.
 
Recall that, for fixed $K > 0$, we define $X_t = \log\frac{S_t}K$ and $Y_t = \alpha_t$, which in the fSABR model satisfy
\bea
&& dX_t = Y_t (\rho dB_t + \bar\rho dW_t) - \frac12 Y_t^2 dt = Y_t d\tW_t - \frac12 Y_t^2 dt, \label{eqn:fSABR-x} \\
&& Y_t = Y_0 e^{\nu B_t^H}. \label{eqn:fSABR-y}
\eea

\begin{remark}
It is easy to see that $Y\in \mathbb{L}^{1,2}_B\cap L^p$, for all $p>1$. On the other hand, by applying Jensen's inequality we have 
$$
\frac{1}{T}\int_0^T Y_s^2 ds \geq  Y_0^2 \exp\left(\frac1{T} \int_0^T 2\nu B_s^H ds\right).
$$
It follows that 
$$
\frac{T}{\int_0^T Y_s^2 ds}\leq  \frac1{Y_0^2} \exp\left(-\frac{2}{T}\int_0^T \nu B_s^H ds\right).
$$
Now, as $\int_0^T \nu B_s^H ds$ is a Gaussian process with the following representation
$$
\int_0^T \nu B_s^H ds=\int_0^T\left(\int_s^T \nu K(s,u)du\right)dB_s,
$$
it follows that $\frac{1}{\int_0^T Y_s^2 ds}\in L^p$, for all $p>1$. In a similar way, one can check that hypotheses {\rm ({\bf H1}), ({\bf H1'}), ({\bf H2})} and {\rm ({\bf H2'})} hold.
\end{remark}

The following lemma summarizes essential quantities that are crucial in the calculation of the approximation formulas \eqref{eqn:approx-tvc} and \eqref{eqn:approx-tvc-tzero} for the price of a TV call in the lognormal fSABR model. 
\begin{lemma} \label{lma:cond-exps}
For $0 < t \leq r \leq T$, define the functions $m$ and $v$ as
\beaa
&& m(r|t) := \int_0^t K(r,s)dB_s, \quad 
v(r|t) := \int_t^r K^2(r,s) ds.
\eeaa
Then, for $0 \leq t < \tau < r < u \leq T$, the following conditional expectations in the lognormal fSABR model can be obtained explicitly.
\beaa
&& \Et{Y_r^2} = Y_0^2 e^{2\nu m(r|t) + 2\nu^2 v(r|t)}, \\
&& \Et{Y_\tau Y_r^2} = Y_0^3 e^{\nu[m(\tau|t) + 2m(r|t)] + \frac{\nu^2}2 \left[v(\tau|t) + 4v(r|t) + 4\int_t^\tau K(\tau,s)K(r,s)ds \right]}, \\
&& \Et{Y_r^2 \mathbb{E}_\tau\left[Y_u^2\right]} = Y_0^4 e^{2\nu[m(r|t) + m(u|t)] + 2\nu^2\left[v(u|\tau) + v(r|t) + \int_t^\tau K^2(u,s)ds + 2\int_t^\tau K(r,s)K(u,s) ds\right]}.
\eeaa
Thus, we have
\beaa
&& \hat w_t = \int_t^T \Et{Y_s^2} ds = Y_0^2 \int_t^T e^{2\nu m(r|t) + 2\nu^2 v(r|t)} dr,  \\
&& M_t = Y_0^2 \int_0^t e^{2\nu B_s^H} ds + Y_0^2 \int_t^T e^{2\nu m(r|t) + 2\nu^2 v(r|t)} dr.
\eeaa
Furthermore, the quadratic variation $\inn{M}$ and covariation $\inn{X, M}$ between $X$ and $M$ are given by 
\beaa
&& d\langle M \rangle_t = 4\nu^2 \left(\int_t^T \Et{Y_r^2} K(r,t) dr\right)^2 dt,  \\
&& d\langle X, M \rangle_t = 2\nu \rho \left(Y_t \int_t^T \Et{Y_r^2} K(r,t) dr\right) dt.
\eeaa
\end{lemma}
\begin{proof}
See the appendix in Section \ref{sec:appendix}.
\end{proof}
We now present the approximation formula \eqref{eqn:approx-tvc} and \eqref{eqn:approx-tvc-tzero}  in the lognormal fSABR model as follows. 
\begin{corollary} (Approximate price of TV call in fSABR) \\
The approximate formula \eqref{eqn:approx-tvc} for the price of the TV call struck at $K$ with expiry $T$ under the lognormal fSABR model \eqref{eqn:fSABR-x}:\eqref{eqn:fSABR-y} is given by 
\bea
&&\Et{\frac1{\sqrt{w_T}} \left( e^{X_T} - 1 \right)^+} \label{eqn:approx-tvc-fsabr} \\
&\approx& \frac1{\sqrt{M_t}} C(X_t, \hat w_t) \nonumber \\
&& + 2 \nu\rho F_{x\hat w}(X_t, w_t, \hat w_t) \, \int_t^T\int_\tau^T \Et{Y_\tau Y_r^2} K(r,\tau) drd\tau \nonumber \\
&& + 4\nu^2 F_{\hat w\hat w}(X_t, w_t, \hat{w}_t) \, \int_t^T\int_\tau^T\int_{r_2}^T \Et{Y_{r_2}^2 \E_{r_2}\left[Y_{r_1}^2\right]}K(r_1, \tau)K(r_2, \tau) dr_1 dr_2 d\tau. \nonumber
\eea
In particular, at $t = 0$ since $m(r|0) = 0$ and $v(r|0) = r^{2H}$, \eqref{eqn:approx-tvc-fsabr} can be expressed more explicitly as 
\bea
&&\Eof{\frac1{\sqrt{w_T}} \left( e^{X_T} - 1 \right)^+} \label{eqn:approx-tvc-fsabr-tzero} \\
&\approx& \frac1{\sqrt{M_0}} C(X_0, M_0) \nonumber \\
&+& 2 \nu\rho F_{x\hat w}(X_0, 0, M_0) \, Y_0^3 \int_0^T\int_\tau^T e^{\frac{\nu^2}2\left[\tau^{2H} + 4r^{2H} + 4 \int_0^\tau K(\tau,s)K(r,s) ds\right]} K(r,\tau) drd\tau \nonumber \\
&+& 4\nu^2 F_{\hat w\hat w}(X_0, 0, M_0) Y_0^4 \times \nonumber \\
&& \int_0^T\int_\tau^T\int_{r_2}^T e^{2\nu^2 \left[v(r_1|\tau) + r_2^{2H} + \int_0^\tau K^2(r_1,s) ds + 2\int_0^\tau K(r_2,s)K(r_1,s) ds\right]} K(r_1, \tau)K(r_2, \tau) dr_1 dr_2 d\tau, \nonumber
\eea
where $M_0 = Y_0^2 \int_0^T e^{2\nu^2 t^{2H}} dt$.
\end{corollary}

\begin{remark}
The numerical implementation of the approximation formula \eqref{eqn:approx-tvc-fsabr} and \eqref{eqn:approx-tvc-fsabr-tzero} is computationally intensive due to the evaluations of the multiple integrals that are involved in the formula.
We further approximate and simplify \eqref{eqn:approx-tvc-fsabr} as follows. 
\bea
&&\Et{\frac1{\sqrt{w_T}} \left( e^{X_T} - 1 \right)^+} \label{eqn:approx-tvc-fsabr-1} \\
&\approx& \frac1{\sqrt{M_t}} C(X_t, \hat w_t) %\nonumber \\
+ 2 \nu\rho F_{x\hat w}(X_t, w_t, \hat w_t) Y_t^3\, \int_t^T\int_\tau^T K(r,\tau) drd\tau \nonumber \\
&& + 4\nu^2 F_{\hat w\hat w}(X_t, w_t, \hat{w}_t) \, Y_t^4\int_t^T\int_\tau^T\int_{r_2}^T  K(r_1, \tau)K(r_2, \tau) dr_1 dr_2 d\tau. \nonumber
\eea
In particular, if $t=0$, we calculate the two integrals on the right hand side of \eqref{eqn:approx-tvc-fsabr-1} as 
\beaa
&& \int_t^T\int_\tau^T K(r,\tau) drd\tau 
= \int_0^T \int_0^r K(r,\tau) d\tau dr = \kappa_H \int_0^T r^{\frac12 + H} dr = \frac{\kappa_H}{\frac32 + H} T^{\frac32 + H}, 
\eeaa
where $\kappa_H := c_H \frac{\beta\left(\frac32 - H, H + \frac12\right)}{H + \frac12}$ with $\beta(\cdot,\cdot)$ being the beta function (see for instance Lemma 2.3 in \cite{asw1} for a proof of the second equality), 
and
\begin{eqnarray*}
&&\int_0^T \int_{\tau}^T \int_{r_2}^T K(r_1,\tau)K(r_2,\tau)dr_1 dr_2 d\tau \\
&=& \int_0^T \int_{r_2}^T \left(    \int_0^{r_2} K(r_1,\tau)K(r_2,\tau)d\tau  \right)dr_1 dr_2 \\
&=& \frac12 \int_0^T \int_{r_2}^T \left(     r_1^{2H}+r_2^{2H}-(r_1-r_2)^{2H}\right)dr_1 dr_2 \\
&=& \frac{T^{2(1 + H)}}{4(1 + H)} 
\end{eqnarray*}
since $\int_0^{r_2} K(r_1,\tau)K(r_2,\tau)d\tau$ is equal to the expectation of $B_{r_2}B_{r_1}$.
It follows that
\bea
&&\Eof{\frac1{\sqrt{w_T}} \left( e^{X_T} - 1 \right)^+} \label{eqn:approx-tvc-fsabr-1-t0} \\
&\approx& \frac1{\sqrt{M_0}} C(X_0, M_0) %\nonumber \\
+ 2 \nu\rho F_{x\hat w}(X_0, 0, M_0) Y_0^3\, \frac{\kappa_H}{\frac32 + H} T^{\frac32 + H} \nonumber \\
&& + \nu^2 F_{\hat w\hat w}(X_0, 0, M_0) \, Y_0^4 \frac{T^{2(1 + H)}}{(1 + H)}. \nonumber
\eea
It is indeed \eqref{eqn:approx-tvc-fsabr-1-t0} that will be implemented in Section \ref{sec:numerics}.  
\end{remark}

%The numerical implementation of the approximation formula \eqref{eqn:approx-tvc-fsabr} and \eqref{eqn:approx-tvc-fsabr-tzero} is computationally intensive due to the evaluations of the multiple integrals that are involved in the formula.
We derive in the next section the small volatility of volatility approximation for the price of a TV call in the lognormal fSABR model, which is numerically more tractable and in a sense can be regarded as a further approximation of the approximation formula \eqref{eqn:approx-tvc-fsabr-tzero}.

%%%%%%%%%%%%%%%%%%%%%%%%%%%%%%%%%%%%%%%%%%%%%%%%%%%%%%%%%%%%%%%%%
%  Section: Small vol of vol expansion
%%%%%%%%%%%%%%%%%%%%%%%%%%%%%%%%%%%%%%%%%%%%%%%%%%%%%%%%%%%%%%%%%
\section{Small volatility of volatility expansion} \label{sec:small-vol-vol}
By conditioning on the $\sigma$-algebra generated by the volatility process up to the expiry of a TV call, we show in this section an alternative approach to the pricing of a TV call. This alternative approach induces an asymptotic of the price of a TV call in the small volatility of volatility regime which is computationally more tractable than \eqref{eqn:approx-tvc-fsabr-tzero} by sacrificing some accuracy.     

To be more specific, recall that the price at time $t=0$ of a TV call struck at $K$ and expiry $T$ is given by the expectation in \eqref{eqn:tvo-price-t} which we recast in the following as 
\bea
K \bar\sigma \sqrt T \, \Eof{\frac1{\sqrt{\int_0^T Y_t^2 dt}} \Eof{\left. \left( e^{X_T}- 1 \right)^+ \right|\cF_T^Y}}, \label{eqn:tvc-condYT}
\eea
where $\cF_T^Y$ is the $\sigma$-algebra generated by the process $Y$ up to time $T$. Of course, $\cF_T^Y$ is equivalent to the $\sigma$-algebra $\cF_T^B$ generated by the Brownian motion $B$ up to time $T$. We shall again temporarily ignore the constant factor outfront in \eqref{eqn:tvc-condYT} in the calculations that follow. 

Notice that from \eqref{eqn:fSABR-x} we have
\beaa
&& X_T = X_0 + \rho\int_0^T Y_t dB_t - \frac12 \int_0^T Y_t^2 dt + \bar\rho \int_0^T Y_t dW_t
\eeaa
and since the two Brownian motions $B_t$ and $W_t$ are independent, 
$X_T|\cF^Y_T$ is normally distributed with mean $\mu_T$ and variance $v_T$ given respectively by
\beaa
&& \mu_T = X_0 + \rho\int_0^T Y_t dB_t - \frac12 \int_0^T Y_t^2 dt, \quad 
v_T = \bar\rho^2 \int_0^T Y_t^2 dt.
\eeaa
Thus, the inner conditional expectation in \eqref{eqn:tvc-condYT} can be evaluated in terms of the Black-Scholes function $C$ defined in \eqref{eqn:BS-func} as
\beaa
&& \Eof{\left. \left( e^{X_T}- 1 \right)^+ \right|\cF_T^Y} 
= C\left(\mu_T + \frac{v_T}2, v_T\right).
\eeaa
It follows that the expectation in \eqref{eqn:tvc-condYT} can be rewritten in terms of $C$ as
\begin{equation}
\Eof{\frac1{\sqrt{\int_0^T Y_t^2 dt}}C\left( X_0 + \rho\int_0^T Y_t dB_t - \frac{\rho^2}2 \int_0^T Y_t^2dt, \bar\rho^2\int_0^T Y_t^2 dt \right)}. \label{eqn:main-exp}
\end{equation}
In particular, in the uncorrelated case $\rho = 0$ the TVO price reduces to
\bea
&& \Eof{\frac1{\sqrt{\int_0^T Y_t^2 dt}}C\left( X_0, \int_0^T Y_t^2 dt \right)} \label{uncorrelated}
\eea
which is simply the classical Black-Scholes with independently randomized total variance. 

\begin{remark}
We remark that if one substitute $Y_t$ in \eqref{eqn:main-exp} by its expectation and evaluated the resulting expression, by straightforward calculations, one can show that it recovers the lowest order term in \eqref{eqn:approx-tvc-tzero}, i.e., the term $\frac1{\sqrt{M_0}}C(X_0, M_0)$.
\end{remark}

Recall that $w_T = \int_0^T Y_t^2 dt$ and define $\xi_T$ by 
\[
\xi_T := X_0 + \rho \int_0^T Y_t dB_t - \frac{\rho^2}2 w_T.
\]
We can also rewrite \eqref{eqn:main-exp} more concisely in terms of the function $F$ defined in \eqref{eqn:func-F} as
\begin{equation}
\Eof{F(\xi_T, \rho^2 w_T, \bar\rho^2 w_T)}. \label{eqn:main-exp-inF}
\end{equation}     
In the lognormal fSABR model $Y_t = Y_0 e^{\nu B_t^H}$, we expand $w_T$ and $\xi_T$ in the volatility of volatility parameter $\nu$ as  
\beaa
&& w_T = \sum_{k=0}^\infty \frac{\nu^k}{k!} w_T^{(k)} \quad \mbox{and} \quad 
\xi_T = \sum_{k=0}^\infty \frac{\nu^k}{k!} \xi_T^{(k)}, 
\eeaa
where
\beaa
&& w_T^{(k)} = Y_0^2 \, 2^k \int_0^T (B_t^H)^k dt, \quad \mbox{ for } k \geq 0, \\
&& \xi_T^{(0)} = X_0 + \rho Y_0 B_T - \frac{\rho^2}2 Y_0^2 T, \\
&& \xi_T^{(k)} = \rho Y_0 \int_0^T (B_t^H)^k dB_t - \frac{\rho^2}2 w_T^{(k)}, \quad \mbox{ for } k \geq 1.
\eeaa

With the aid of the identities in Lemmas \ref{lma:normal-expects} and \ref{lma:cond-exps} in Section \ref{sec:appendix}, we now show the small volatility of volatility expansion for the price of a TV call given in \eqref{eqn:main-exp} as $\nu$ approaches zero. 
\begin{theorem}(Small volatility of volatility asympotitcs for TV call price) \label{thm:tv-price-small-volvol} \\
The price of a TV call struck at $K$ and expiry $T$ has the following asymptotic up to the first order as $\nu \to 0$.
\bea
&& \Eof{\frac1{\sqrt w_T}\left( e^{X_T} - 1 \right)^+} \label{eqn:tv-price-small-volvol2} \\
&=& \frac1{Y_0\sqrt T} \Eof{C}
+ \frac{\nu}{Y_0 \sqrt T} \Eof{C_x \xi_T^{(1)} + \bar\rho^2 C_w w_T^{(1)}} 
- \frac\nu{2 Y_0^3 T^{\frac32}} \Eof{w_T^{(1)} C} + O(\nu^2), \nonumber
\eea
where the function $C$ and all its partial derivatives in the last expression are evaluated at $\left(\xi_T^{(0)}, \bar\rho^2 w_T^{(0)} \right)$. 
%The expectations on the right side of \eqref{eqn:tv-price-small-volvol} are given explicitly by 
%\beaa
%&& \Eof{C} = C(X_0, Y_0^2T) \\
%&& \Eof{C_x \xi_T^{(1)}} =  \\
%&& \Eof{C_w w_T^{(1)}} = \\
%&& \Eof{w_T^{(1)} C} = 
%\eeaa
\end{theorem}
\begin{proof}
As $\nu \to 0$, consider \eqref{eqn:main-exp-inF}  
\beaa
%&& \Eof{\frac1{\sqrt{w_T}}C\left(\xi_T , \bar\rho^2 w_T \right)} \\
&& \Eof{F(\xi_T, \rho^2 w_T, \bar\rho^2 w_T)} \\
%&=& \E\left[\frac1{\sqrt{\sum_{k=0}^\infty \frac{\nu^k}{k!} w_T^{(k)}}} \; C\left(\sum_{k=0}^\infty \frac{\nu^k}{k!} \xi_T^{(k)}, \bar\rho^2 \sum_{k=0}^\infty \frac{\nu^k}{k!} w_T^{(k)} \right) \right] \\
&=& \Eof{F\left( \sum_{k=0}^\infty \frac{\nu^k}{k!} \xi_T^{(k)}, \rho^2 \sum_{k=0}^\infty \frac{\nu^k}{k!} w_T^{(k)}, \bar\rho^2 \sum_{k=0}^\infty \frac{\nu^k}{k!} w_T^{(k)} \right)}.
\eeaa
To the first order we have
\beaa
&& \Eof{F(\xi_T^{(0)}, \rho^2 w_T^{(0)}, \bar\rho^2 w_T^{(0)})} + \nu \Eof{F_x \xi_T^{(1)} + F_w \rho^2 w_T^{(1)} + F_{\hat w} \bar\rho^2 w_T^{(1)})} + O(\nu^2) \\
&=& \Eof{\frac C{\sqrt{w_T^{(0)}}}}
+ \nu \E\left[\frac{C_x}{\sqrt{w_T^{(0)}}} \xi_T^{(1)} + \bar\rho^2 \frac{C_w}{\sqrt{w_T^{(0)}}} w_T^{(1)} 
- \frac{C}{2\left(w_T^{(0)}\right)^{3/2}} w_T^{(1)} \right] + O(\nu^2) \\
&=& \frac1{Y_0\sqrt T} \Eof{C}
+ \frac{\nu}{Y_0 \sqrt T} \Eof{C_x \xi_T^{(1)} + \bar\rho^2 C_w w_T^{(1)}} 
- \frac\nu{2 Y_0^3 T^{\frac32}} \Eof{w_T^{(1)} C} + O(\nu^2),
\eeaa
where the function $C$ and all its partial derivatives in the last expression are evaluated at $\left(\xi_T^{(0)}, \bar\rho^2 w_T^{(0)} \right)$. Notice that $\xi_T^{(0)}$ is a random variable, since it is a linear function of $B_T$, whereas $w_T^{(0)} = Y_0^2 T$ is deterministic. 
\end{proof}

The expectations in the last expression of \eqref{eqn:tv-price-small-volvol2} can be obtained explicitly, which leads to an explicit approximation formula for the price of a TV call. The following corollary gives the complete pricing formula for computational purposes. 

\begin{corollary}\label{cor:tvoComputationFormula}
Let $X_0=\log{\frac{S_0}{K}}$ and assume $Y_0$ is given. Then the price of a TV call option expiring at $T>0$ with strike $K$ and target volatility $\bar\sigma$ has the following first order approximation as $\nu\to0$:
\bea
\text{TVO} &:=& K\bar\sigma\sqrt{T}\Eof{\frac1{\sqrt{\int_0^T Y_\tau d\tau}}\Big( e^{X_T} - 1 \Big)^+\Big| (X_0,Y_0) } \nonumber\\
& \approx & \frac{K\bar\sigma}{Y_0}C(X_0,T Y_0^2) + \frac{K \bar\sigma} {2 \sqrt{\pi } (2 H+3) \rho  Y_0^2} \times\nonumber\\
&&\Big[c_H \nu  T^{H-\frac{1}{2}} \exp \Big(-\frac{\rho^2 X_0^2}{T Y_0^2}-\frac{1}{2} \Big(2 \rho ^2+1\Big) T Y_0^2-\frac{Y_0}{2 \sqrt{T}}\Big)\times \nonumber\\
&&\Big(\sqrt{\pi } e^{\frac{\rho ^2 X_0^2}{T Y_0^2}+\frac{Y_0}{2 \sqrt{T}}}
   \Big(N\Big(\frac{T Y_0^2+2 X_0}{2 \sqrt{T} Y_0}\Big) \Big(e^{\frac{1}{2} \Big(2 \rho ^2+1\Big) T Y_0^2+X}\times\nonumber\\
   && \Big(4 T X_0^2+4 X_0 \Big(T Y_0^2
   \Big(\rho ^2-2 \rho ^2 T+T\Big)+1\Big)+\nonumber\\
   && T Y_0^2 \Big(-4 \rho ^2+T \Big(Y_0^2 \Big(4 \rho ^4 (T-1)+2 \rho ^2 (1-2 T)+T\Big)-4 \rho
   ^2\Big)+2\Big)\Big)+\nonumber\\
   && 2 \sqrt{1-\rho ^2} \Big(T Y_0^2+2 X_0\Big) e^{\frac{1}{2} \Big(3 \rho ^2+\sqrt{1-\rho ^2}\Big) T Y_0^2+\sqrt{1-\rho ^2} X_0}\nonumber\\
   &&\Big(T \Big(Y_0^2 \Big(\rho ^2 (2 T-1)-T\Big)-2 X_0\Big)-1\Big)\Big)+\sqrt{T} \Big(\sqrt{T} \times\nonumber\\
   &&e^{\frac{3 T Y_0^2}{2}+X_0} N'\Big(\frac{T Y_0^2+2X_0}{2 \sqrt{T} Y_0}\Big) \times\nonumber\\
   &&\Big(-2 \sqrt{T} X_0 Y_0 \Big(2 \rho ^4-4  \rho ^2+\Big(\rho ^2-2\Big) \sqrt{T}
   Y_0\Big)+T Y_0^2\times\nonumber\\
   &&
   \Big(2 \rho ^2+\Big(\rho ^2-1\Big) \Big(6 \sqrt{\frac{1}{1-\rho ^2}} \rho ^4+1\Big) (-T) Y_0^2-2 \rho ^2 \Big(5 \rho ^2-2 \Big) \sqrt{T} Y_0\Big)+4 X_0^2\Big)+\nonumber\\
   &&
   4 \rho ^2 Y_0
   N'\Big(\frac{X_0}{\sqrt{T} Y_0}-\frac{\sqrt{T} Y_0}{2}\Big) \Big(e^{\frac{1}{2} \Big(2 \rho ^2+1\Big) T Y_0^2}-e^{\frac{1}{2} \Big(3 \rho
   ^2+\sqrt{1-\rho ^2}\Big) T Y_0^2+\sqrt{1-\rho ^2} X_0} \times\nonumber\\
   &&
   \Big(T \Big(Y_0^2 \Big(\rho ^2-2 \rho ^2 T+T\Big)+2
   X_0\Big)+1\Big)\Big)\Big)\Big)+\sqrt{2} \Big(\rho ^2-1\Big) \rho ^2 \sqrt{T} Y_0 \Big(T Y_0^2-2 X\Big)\times\nonumber\\
   && \exp \Big(\frac{X_0}{T^{3/2}Y_0}+\frac{1}{4} \Big(3 \rho ^2+2\Big) T Y_0^2+\rho ^2 X_0\Big)\Big)\Big]\label{eqn:tv-price-small-volvol}
\eea 
\end{corollary}
\begin{proof}
We first calculate the zeroth order term as follows. 
\beaa
&& \Eof{C} = \Eof{C\left(\xi_T^{(0)}, \bar\rho^2 w_T^{(0)} \right)} \\
%&=& \Eof{e^{\xi_T^{(0)}} N(d_1)} - \Eof{N(d_2)} \\
&=& \Eof{e^{\xi_T^{(0)}} N\left( \frac{\xi_T^{(0)}}{\sqrt{\bar\rho^2 w_T^{(0)}}} + \frac{\sqrt{\bar\rho^2 w_T^{(0)}}}2 \right)} - \Eof{N\left( \frac{\xi_T^{(0)}}{\sqrt{\bar\rho^2 w_T^{(0)}}} - \frac{\sqrt{\bar\rho^2 w_T^{(0)}}}2 \right)} \\
&=& e^{X_0} N\left( \frac{X_0}{\sqrt{Y_0^2T}} + \frac{\sqrt{Y_0^2T}}2 \right) - N\left( \frac{X_0}{\sqrt{Y_0^2T}} - \frac{\sqrt{Y_0^2T}}2 \right) \\
&=& C(X_0, Y_0^2 T),
\eeaa
which unsurprisingly is independent of $\rho$ and $\nu$. Note that in passing to the last equality, we used \eqref{eqn:E-NX} and \eqref{eqn:E-eXNX}. The rest of the calculations, though straightforward, are more involved and tedious. We leave the details to the Appendix, see section \ref{sec:proof-of-cor}.
\end{proof}
We remark that, up to first order of $\nu$, the first three terms in the last expression represent the price of a vanilla option scale up/down by the factor $\frac{\bar\sigma}{Y_0}$, see Theorem \ref{thm:small-vol-vol-vanilla}. The last term corresponds to a correction to vanilla due to the uncertainty of realized variance.

%%%%%%%%%%%%%%%%%%%%%%%%%%%%%%%%%%%%%%%%%%%%%%%%%%%%%%%%%%%%%%%%%
%  Subsection: Expansion for price of vanilla option
%%%%%%%%%%%%%%%%%%%%%%%%%%%%%%%%%%%%%%%%%%%%%%%%%%%%%%%%%%%%%%%%%
\subsection{Expansion for price of vanilla option}
We consider the small volatility of volatility expansion of vanilla calls in this subsection. To our knowledge, the small volatility of volatility expansion for vanilla option under fSABR model does not seem to exist in literature by the time the paper was written up. 

The premium of a vanilla call struck at $K$ and expiry $T$ in our notation is given by the expectation under risk neutral probability 
\bea
&& K \, \Eof{\left(e^{X_T} - 1 \right)^+} \nonumber \\
&=& K \, \Eof{C\left( X_0 + \rho\int_0^T Y_t dB_t - \frac{\rho^2}2 \int_0^T Y_t^2dt, \bar\rho^2\int_0^T Y_t^2 dt \right)}, \label{eqn:vanilla-call}
\eea
where $C$ is again the normalized Black-Scholes function defined in \eqref{eqn:BS-func}. 
\begin{theorem}(Small volatility of volatility asympotitc for vanilla call) \label{thm:small-vol-vol-vanilla} \\
The price of a vanilla call struck at $K$ and expiry $T$ in the lognormal fSABR model has the following asymptotic 
\beaa
&& K \, \Eof{\left(e^{X_T} - 1 \right)^+}  \\
&\approx& K C(X_0, Y_0^2T) + \nu \frac{2c_H K}{2H + 3} T^{H + \frac12} \Eof{2 Y_0^2 C_w B_T + \rho Y_0 C_x \left\{ B_T^2 - T - \rho Y_0 B_T \right\}} \\ 
&& + O(\nu^2).
\eeaa
as the volatility of volatility parameter $\nu$ approaches zero.
\end{theorem}
\begin{proof}
Temporarily ignoring the constant $K$, \eqref{eqn:vanilla-call} can be written in terms of $\xi_T$ and $w_T$ as 
\beaa
&& \Eof{C\left( \xi_T, \bar\rho^2 w_T \right)} 
= \E\left[C\left(\sum_{k=0}^\infty \frac{\nu^k}{k!} \xi_T^{(k)}, \bar\rho^2 \sum_{k=0}^\infty \frac{\nu^k}{k!} w_T^{(k)} \right)\right].
\eeaa
Thus, up to first order we have
\beaa
&& \Eof{C\left(\xi_T, \bar\rho^2 w_T \right)} 
= \Eof{C} + \nu \Eof{C_x \xi_T^{(1)} + \bar\rho^2 C_w w_T^{(1)}} + O(\nu^2),
\eeaa
where $C$ and all its partial derivatives on the right hand side of the last equation are evaluated at $\left(\xi_T^{(0)}, \bar\rho^2 w_T^{(0)} \right)$. The calculations of the expectations in the last expression are the same as the ones for Theorem \ref{thm:tv-price-small-volvol}, see Section \ref{sec:proof-of-cor} in the appendix for the details. We remark that theoretically it is possible to push for higher order terms in the small volatility of volatility expansion, however, the calculation of expectations becomes more and more involved for higher order terms. 
\end{proof}

\section{Numerical implementation} \label{sec:numerics}
We present in this section several simulation results, showing the applicability of our pricing formulae for Target Volatility call Options under the fSABR model. First of all, note that we classify model's parameters into three categories: (i) contract--specific parameters: $K, T, \bar\sigma$; (ii) model--specific parameters: $H,\nu,\rho$, and (iii) initialization parameters: $S_0, \sigma_0, n, N$. Here, $n\in\mathbb{N}$ is the number of sampling days per annum, $N\in\mathbb{N}$ is the number of Monte Carlo paths, and $\sigma_0=Y_0$ is the volatility at time $t=0$.  Moreover, recall that $X_0:=-\log\frac{K}{S_0}$.

We employ for verification purposes three pricing methods: Monte Carlo Simulations, Decomposition Formula Approximation \eqref{eqn:approx-tvc-fsabr-1-t0}, and Small Volatility of Volatility Expansion \eqref{eqn:tv-price-small-volvol}. Observe that both formulas \eqref{eqn:approx-tvc-fsabr-1-t0} and \eqref{eqn:tv-price-small-volvol} can be easily implemented numerically as they only require the use of special functions: the pdf $N'$ and cdf $N$ for a standard normal distribution, the Euler Gamma function $\Gamma$, the Gauss hypergeometric function ${ }_2F_1$, and the Beta function $\beta$. As we shall see in the sequel, our approximation formulas perform accurately for a wide range of parameters. 

Sample paths for fractional Brownian Motions $\{ B^H(t_k),\ k=1,2,\dots,n\}$ using the Molchan--Golosov kernel are simulated. Here, we consider a partition $\Pi:=\{0 = t_0<t_1<\cdots<t_n=T\}$ of the interval $[0,T]$. We employ the hybrid scheme for Brownian semistationary processes given in \cite{bennedsen2017hybrid}, which is based on discretizing the stochastic integral representation of the process in the time domain. Several test routines for fractional processes are also implemented: mean and variance as a function of time via Monte Carlo simulations, a chi--square test for fractional Gaussian noise, as well as the 2D correlation structure via sample paths. We notice that the sample paths have the required properties, that are specific to fBMs.

%More recently, \cite{mccrickerd2017turbocharging} considered methods for turbocharging Monte Carlo pricing under the so--called rough Bergomi model. 

%Several test routines for fractional processes have been implemented: mean and variance as a function of time via Monte Carlo simulations, a chi--square test for fractional Gaussian noise \cite{dieker2004simulation}, as well as the 2D correlation structure via sample paths: 
%\bea
%\mu(t_k) &=& \frac{1}{N} \sum_{j=1}^N B^H_j (t_k),\\
%v(t_k) &=& \frac{1}{N-1} \sum_{j=1}^N \Big(B^H_j (t_k) - \mu(t_k)\Big)^2,
%\eea
%with $k = 1,2,\dots,n$. The results are presented in Figure \ref{fig:evar} and Figure \ref{fig:covFBM}. It is easy to notice that the sample paths have the required properties, that are specific to fBMs.   

\subsection{Formula accuracy}
To test the accuracy of  Decomposition Formula Approximation \eqref{eqn:approx-tvc-fsabr-1-t0}, and Small Volatility of Volatility Expansion \eqref{eqn:tv-price-small-volvol}, we produce sample paths for the lognormal fSABR price process and we use Monte Carlo techniques to calculate TVO prices. Let $\Pi$ be as before. Then one path of the lognormal fSABR price process can be computed iteratively:
\[
S(t_{k+1}) = S(t_k)+ S(t_k)\sigma(t_k)\varepsilon(t_k)\sqrt{t_{k+1} - t_{k}},\quad k=0,1,\dots,n-1,
\]
where $\sigma(t) = \sigma_0e^{\nu B^H(t)}$ is the so--called fractional stochastic volatility and the random samples $\varepsilon_{t_i} = \rho\xi^1_{t_i} + \bar\rho\xi^2_{t_i},\ \xi^{1,2}_{t_i}\sim\mathcal{N}(0,1),\ i=0,1,\dots,n-1$ have a standard normal distribution. 

Numerical results for $\frac{K}{S_0} \in [0.8,1.2]$, different maturities, and a wide range of parameters are presented in Tables \ref{tab:t1} to \ref{tab:t6} and Figures \ref{fig:f1} to \ref{fig:f5}. The relative error between the Monte Carlo prices and the pricing formulas \eqref{eqn:approx-tvc-fsabr-1-t0}, \eqref{eqn:tv-price-small-volvol} is also calculated for illustrative purposes. We identify several advantages and disadvantages for using one formula or the other:
\begin{enumerate}[(i)]
\item Formula \eqref{eqn:tv-price-small-volvol} is the most computationally efficient method (3 times faster on average than its counterpart). It can be noticed that it works well for small volatility of volatility, namely $0<\nu<15\%$, larger times to maturity, and $\frac{K}{S_0}$ close to 1. Also note from the tables that \eqref{eqn:tv-price-small-volvol} is more accurate for small values of $\nu$ than its counterpart \eqref{eqn:approx-tvc-fsabr-1-t0}.
\item Formula \eqref{eqn:approx-tvc-fsabr-1-t0} is a better approximation for options that are far ITM or OTM. From our tests, it works well and is robust with almost any parameters.  
\end{enumerate}
Thus, we recommend using \eqref{eqn:tv-price-small-volvol} for computational efficiency, when $\nu$ is small and $T>0.5$ years. For all other purposes, \eqref{eqn:approx-tvc-fsabr-1-t0} is the better choice. Moreover, although we only show results for the rough regime, i.e. $H\in(0,0.5)$, we note that both formulas work well for all possible values of the Hurst exponent $H\in(0,1)$. Clearly, our formulas are less accurate when $H$ is close to 0 or 1.

\subsection{Sensitivity to parameters}
In order to stress test our formulas, we compute the TVO price At-The-Money via \eqref{eqn:approx-tvc-fsabr-1-t0}, \eqref{eqn:tv-price-small-volvol}, and Monte Carlo simulations for a broad range of parameters $(H,\nu,\rho)$. Namely, we consider $H\in(0,0.5),\ \nu\in(0,0.6),$ and $\rho\in(-1,1)$. The results are presented in Figures \ref{fig:f7} to \ref{fig:f12}. Firstly, we plot the TVO price as a function of 2 parameters while assuming the 3rd being fixed. Secondly, we compute and plot the relative error between our formulas and prices via Monte Carlo trials. Note that the relative error is small, and that the price surfaces are fairly smooth. We emphasize one more time that the approximation formulas turns out to be highly accurate and robust to parameter variations.

\begin{table}
{\scriptsize
\centering
\begin{filecontents*}{tableTex_1r.csv}
0.887 0.181 0.182 0.179 0.562 0.830
0.896 0.176 0.177 0.174 0.565 0.766
0.905 0.170 0.171 0.169 0.565 0.701
0.914 0.165 0.166 0.164 0.568 0.630
0.923 0.160 0.161 0.159 0.576 0.551
0.932 0.155 0.156 0.154 0.585 0.467
0.942 0.150 0.151 0.149 0.593 0.381
0.951 0.145 0.146 0.144 0.600 0.292
0.961 0.140 0.141 0.140 0.609 0.199
0.970 0.135 0.136 0.135 0.618 0.101
0.980 0.130 0.131 0.130 0.631 0.003
0.990 0.126 0.126 0.126 0.640 0.106
1.000 0.121 0.122 0.121 0.645 0.208
1.010 0.117 0.117 0.117 0.652 0.316
1.020 0.112 0.113 0.113 0.658 0.425
1.030 0.108 0.109 0.108 0.665 0.538
1.041 0.104 0.104 0.104 0.681 0.663
1.051 0.100 0.100 0.100 0.698 0.791
1.062 0.096 0.096 0.096 0.713 0.919
1.073 0.092 0.092 0.093 0.732 1.053
1.083 0.088 0.088 0.089 0.756 1.194
1.094 0.084 0.085 0.085 0.779 1.337
1.105 0.080 0.081 0.082 0.798 1.476
1.116 0.077 0.077 0.078 0.820 1.619
\end{filecontents*}
\pgfplotstabletypeset[%
   fixed zerofill,
   precision=4,
   col sep=space,
   %dec sep align,
   every head row/.style={before row=\hline,after row=\hline},
   every last row/.style={after row=\hline},
   columns/0/.style ={column name= $K/S_0$},
   columns/1/.style ={column name=MC simulation},
   columns/2/.style ={column name=DFA \eqref{eqn:approx-tvc-fsabr-1-t0}},
   columns/3/.style ={column name=SVVE \eqref{eqn:tv-price-small-volvol}},
   columns/4/.style ={column name=DFA rel.err. (\%)},
   columns/5/.style ={column name=SVVE rel.err. (\%)},
   columns/.style={column type={c}},
]{tableTex_1r.csv}}
\caption{\footnotesize$T = 1$ year, $\bar\sigma = \sigma_0 = 0.3, H = 0.1, \nu = 0.05, \rho = -0.7, n=252, N=50,000$.}
\label{tab:t1}
\end{table}

%%%%%%%%%%%%%%%%%%%%%%
\begin{table}
{\scriptsize
\centering
\begin{filecontents*}{tableTex_2r.csv}
0.887 0.187 0.187 0.198 0.364 5.593
0.896 0.177 0.176 0.186 0.375 5.302
0.905 0.167 0.166 0.175 0.383 4.996
0.914 0.157 0.157 0.165 0.399 4.659
0.923 0.148 0.147 0.154 0.420 4.293
0.932 0.139 0.138 0.144 0.439 3.902
0.942 0.130 0.129 0.134 0.467 3.471
0.951 0.121 0.120 0.125 0.494 3.011
0.961 0.113 0.112 0.115 0.536 2.502
0.970 0.105 0.104 0.107 0.591 1.947
0.980 0.097 0.096 0.098 0.659 1.344
0.990 0.090 0.089 0.090 0.736 0.698
1.000 0.083 0.082 0.083 0.810 0.024
1.010 0.076 0.075 0.076 0.901 0.698
1.020 0.070 0.069 0.069 1.029 1.485
1.030 0.064 0.063 0.063 1.172 2.310
1.041 0.059 0.058 0.057 1.320 3.157
1.051 0.053 0.053 0.051 1.450 3.999
1.062 0.048 0.048 0.046 1.544 4.812
1.073 0.044 0.043 0.041 1.649 5.635
1.083 0.040 0.039 0.037 1.676 6.373
1.094 0.036 0.035 0.033 1.692 7.083
1.105 0.032 0.031 0.030 1.704 7.760
1.116 0.029 0.028 0.026 1.704 8.385
\end{filecontents*}
\pgfplotstabletypeset[%
   fixed zerofill,
   precision=4,
   col sep=space,
   %dec sep align,
   every head row/.style={before row=\hline,after row=\hline},
   every last row/.style={after row=\hline},
   columns/0/.style ={column name= $K/S_0$},
   columns/1/.style ={column name=MC simulation},
   columns/2/.style ={column name=DFA \eqref{eqn:approx-tvc-fsabr-1-t0}},
   columns/3/.style ={column name=SVVE \eqref{eqn:tv-price-small-volvol}},
   columns/4/.style ={column name=DFA rel.err. (\%)},
   columns/5/.style ={column name=SVVE rel.err. (\%)},
   columns/.style={column type={c}},
]{tableTex_2r.csv}}
\caption{\footnotesize$T = 0.5$ years, $\bar\sigma = 0.3, \sigma_0 = 0.2, H = 0.2, \nu = 0.1, \rho = 0.5, n=126, N=50,000$.}
\label{tab:t2}
\end{table}

%%%%%%%%%%%%%%%%%%%%%%
\begin{table}
{\scriptsize
\centering
\begin{filecontents*}{tableTex_3r.csv}
0.887 0.059 0.059 0.059 0.122 0.275
0.896 0.055 0.055 0.055 0.148 0.050
0.905 0.051 0.052 0.052 0.177 0.090
0.914 0.048 0.048 0.048 0.210 0.180
0.923 0.044 0.044 0.044 0.251 0.242
0.932 0.041 0.041 0.041 0.290 0.281
0.942 0.037 0.037 0.037 0.327 0.310
0.951 0.034 0.034 0.034 0.347 0.319
0.961 0.031 0.031 0.031 0.372 0.333
0.970 0.028 0.028 0.028 0.403 0.360
0.980 0.025 0.025 0.025 0.428 0.386
0.990 0.022 0.022 0.022 0.454 0.419
1.000 0.020 0.020 0.020 0.475 0.454
1.010 0.018 0.018 0.018 0.518 0.516
1.020 0.015 0.015 0.016 0.582 0.601
1.030 0.013 0.014 0.014 0.687 0.726
1.041 0.012 0.012 0.012 0.808 0.862
1.051 0.010 0.010 0.010 0.944 1.004
1.062 0.009 0.009 0.009 1.168 1.216
1.073 0.007 0.007 0.007 1.409 1.422
1.083 0.006 0.006 0.006 1.695 1.641
1.094 0.005 0.005 0.005 1.992 1.827
1.105 0.004 0.004 0.004 2.321 1.992
1.116 0.004 0.004 0.004 2.652 2.092
\end{filecontents*}
\pgfplotstabletypeset[%
   fixed zerofill,
   precision=4,
   col sep=space,
   %dec sep align,
   every head row/.style={before row=\hline,after row=\hline},
   every last row/.style={after row=\hline},
   columns/0/.style ={column name= $K/S_0$},
   columns/1/.style ={column name=MC simulation},
   columns/2/.style ={column name=DFA \eqref{eqn:approx-tvc-fsabr-1-t0}},
   columns/3/.style ={column name=SVVE \eqref{eqn:tv-price-small-volvol}},
   columns/4/.style ={column name=DFA rel.err. (\%)},
   columns/5/.style ={column name=SVVE rel.err. (\%)},
   columns/.style={column type={c}},
]{tableTex_3r.csv}}
\caption{\footnotesize$T = 0.25$ years, $\bar\sigma = 0.1, \sigma_0 = 0.2, H = 0.2, \nu = 0.01, \rho = -0.1, n=1000, N=50,000$.}
\label{tab:t3}
\end{table}

%%%%%%%%%%%%%%%%%%%%%%
\begin{table}
{\scriptsize
\centering
\begin{filecontents*}{tableTex_4r.csv}
0.887 0.106 0.107 0.227 0.415 114.341
0.896 0.097 0.098 0.206 0.491 111.106
0.905 0.089 0.089 0.184 0.600 106.812
0.914 0.080 0.081 0.162 0.741 101.622
0.923 0.072 0.072 0.140 0.911 95.665
0.932 0.063 0.064 0.120 1.106 89.001
0.942 0.056 0.056 0.101 1.284 81.517
0.951 0.048 0.049 0.083 1.440 73.065
0.961 0.041 0.042 0.068 1.546 63.358
0.970 0.035 0.036 0.053 1.627 52.126
0.980 0.029 0.030 0.041 1.645 38.992
0.990 0.024 0.025 0.030 1.567 23.658
1.000 0.020 0.020 0.021 1.417 6.057
1.010 0.016 0.016 0.014 1.244 13.638
1.020 0.013 0.013 0.009 1.105 34.982
1.030 0.010 0.011 0.004 1.133 57.215
1.041 0.008 0.008 0.002 1.307 79.427
1.051 0.007 0.007 -0.000 1.809 100.501
1.062 0.005 0.005 -0.001 2.443 119.359
1.073 0.004 0.004 -0.001 3.199 135.006
1.083 0.003 0.003 -0.001 3.962 146.730
1.094 0.002 0.002 -0.001 4.708 154.252
1.105 0.002 0.002 -0.001 4.532 157.230
1.116 0.001 0.001 -0.001 3.429 156.158
\end{filecontents*}
\pgfplotstabletypeset[%
   fixed zerofill,
   precision=4,
   col sep=space,
   %dec sep align,
   every head row/.style={before row=\hline,after row=\hline},
   every last row/.style={after row=\hline},
   columns/0/.style ={column name= $K/S_0$},
   columns/1/.style ={column name=MC simulation},
   columns/2/.style ={column name=DFA \eqref{eqn:approx-tvc-fsabr-1-t0}},
   columns/3/.style ={column name=SVVE \eqref{eqn:tv-price-small-volvol}},
   columns/4/.style ={column name=DFA rel.err. (\%)},
   columns/5/.style ={column name=SVVE rel.err. (\%)},
   columns/.style={column type={c}},
]{tableTex_4r.csv}}
\caption{\footnotesize$T = 0.33$ years, $\bar\sigma = \sigma_0 = 0.1, H = 0.2, \nu = 0.3, \rho = 0.8, n=1000, N=50,000$.}
\label{tab:t4}
\end{table}

%%%%%%%%%%%%%%%%%%%%%%
\begin{table}
{\scriptsize
\centering
\begin{filecontents*}{tableTex_5r.csv}
0.887 0.351 0.350 0.272 0.259 22.576
0.896 0.326 0.325 0.254 0.273 21.960
0.905 0.301 0.300 0.237 0.286 21.188
0.914 0.276 0.275 0.220 0.305 20.251
0.923 0.252 0.251 0.204 0.334 19.135
0.932 0.228 0.227 0.187 0.367 17.817
0.942 0.205 0.204 0.172 0.402 16.277
0.951 0.183 0.182 0.156 0.442 14.497
0.961 0.161 0.161 0.141 0.492 12.462
0.970 0.141 0.140 0.127 0.571 10.169
0.980 0.122 0.121 0.113 0.649 7.574
0.990 0.105 0.104 0.100 0.761 4.694
1.000 0.089 0.088 0.087 0.898 1.514
1.010 0.074 0.073 0.076 1.060 1.975
1.020 0.061 0.060 0.065 1.255 5.762
1.030 0.050 0.049 0.055 1.424 9.906
1.041 0.040 0.039 0.046 1.651 14.305
1.051 0.032 0.031 0.037 2.064 18.782
1.062 0.025 0.024 0.030 2.653 23.299
1.073 0.019 0.018 0.024 3.217 28.063
1.083 0.014 0.014 0.019 3.740 33.065
1.094 0.010 0.010 0.014 4.515 37.845
1.105 0.008 0.007 0.011 5.717 42.038
1.116 0.005 0.005 0.008 6.857 46.263
\end{filecontents*}
\pgfplotstabletypeset[%
   fixed zerofill,
   precision=4,
   col sep=space,
   %dec sep align,
   every head row/.style={before row=\hline,after row=\hline},
   every last row/.style={after row=\hline},
   columns/0/.style ={column name= $K/S_0$},
   columns/1/.style ={column name=MC simulation},
   columns/2/.style ={column name=DFA \eqref{eqn:approx-tvc-fsabr-1-t0}},
   columns/3/.style ={column name=SVVE \eqref{eqn:tv-price-small-volvol}},
   columns/4/.style ={column name=DFA rel.err. (\%)},
   columns/5/.style ={column name=SVVE rel.err. (\%)},
   columns/.style={column type={c}},
]{tableTex_5r.csv}}
\caption{\footnotesize$T = 0.5$ years, $\bar\sigma = 0.3, \sigma_0 = 0.1, H = 0.3, \nu = 0.1, \rho = -0.7, n=1000, N=50,000$.}
\label{tab:t5}
\end{table}

%%%%%%%%%%%%%%%%%%%%%%
\begin{table}
{\scriptsize
\centering
\begin{filecontents*}{tableTex_6r.csv}
0.887 0.131 0.129 0.124 0.893 5.268
0.896 0.123 0.122 0.116 0.932 5.923
0.905 0.116 0.115 0.109 0.971 6.404
0.914 0.109 0.108 0.102 1.013 6.700
0.923 0.102 0.101 0.095 1.058 6.799
0.932 0.095 0.094 0.089 1.105 6.690
0.942 0.088 0.087 0.083 1.155 6.369
0.951 0.082 0.081 0.077 1.202 5.828
0.961 0.076 0.075 0.072 1.243 5.065
0.970 0.070 0.069 0.067 1.287 4.097
0.980 0.064 0.063 0.062 1.313 2.918
0.990 0.058 0.057 0.057 1.329 1.554
1.000 0.053 0.052 0.053 1.348 0.047
1.010 0.048 0.047 0.049 1.343 1.597
1.020 0.043 0.042 0.044 1.340 3.305
1.030 0.039 0.038 0.040 1.339 5.023
1.041 0.034 0.034 0.037 1.308 6.720
1.051 0.030 0.030 0.033 1.242 8.322
1.062 0.027 0.027 0.029 1.141 9.740
1.073 0.024 0.023 0.026 0.993 10.876
1.083 0.020 0.020 0.023 0.846 11.550
1.094 0.018 0.018 0.020 0.665 11.648
1.105 0.015 0.015 0.017 0.408 11.040
1.116 0.013 0.013 0.014 0.052 9.540
\end{filecontents*}
\pgfplotstabletypeset[%
   fixed zerofill,
   precision=4,
   col sep=space,
   %dec sep align,
   every head row/.style={before row=\hline,after row=\hline},
   every last row/.style={after row=\hline},
   columns/0/.style ={column name= $K/S_0$},
   columns/1/.style ={column name=MC simulation},
   columns/2/.style ={column name=DFA \eqref{eqn:approx-tvc-fsabr-1-t0}},
   columns/3/.style ={column name=SVVE \eqref{eqn:tv-price-small-volvol}},
   columns/4/.style ={column name=DFA rel.err. (\%)},
   columns/5/.style ={column name=SVVE rel.err. (\%)},
   columns/.style={column type={c}},
]{tableTex_6r.csv}}
\caption{\footnotesize$T = 1.5$ years, $\bar\sigma = \sigma_0 = 0.1, H = 0.2, \nu = 0.2, \rho = -0.5, n=1000, N=50,000$.}
\label{tab:t6}
\end{table}

%%%%%%%%%%%%%%%%%%%%%%Figures%%%%%%%%%%%%

\begin{figure}
\centering
\includegraphics[height=5cm]{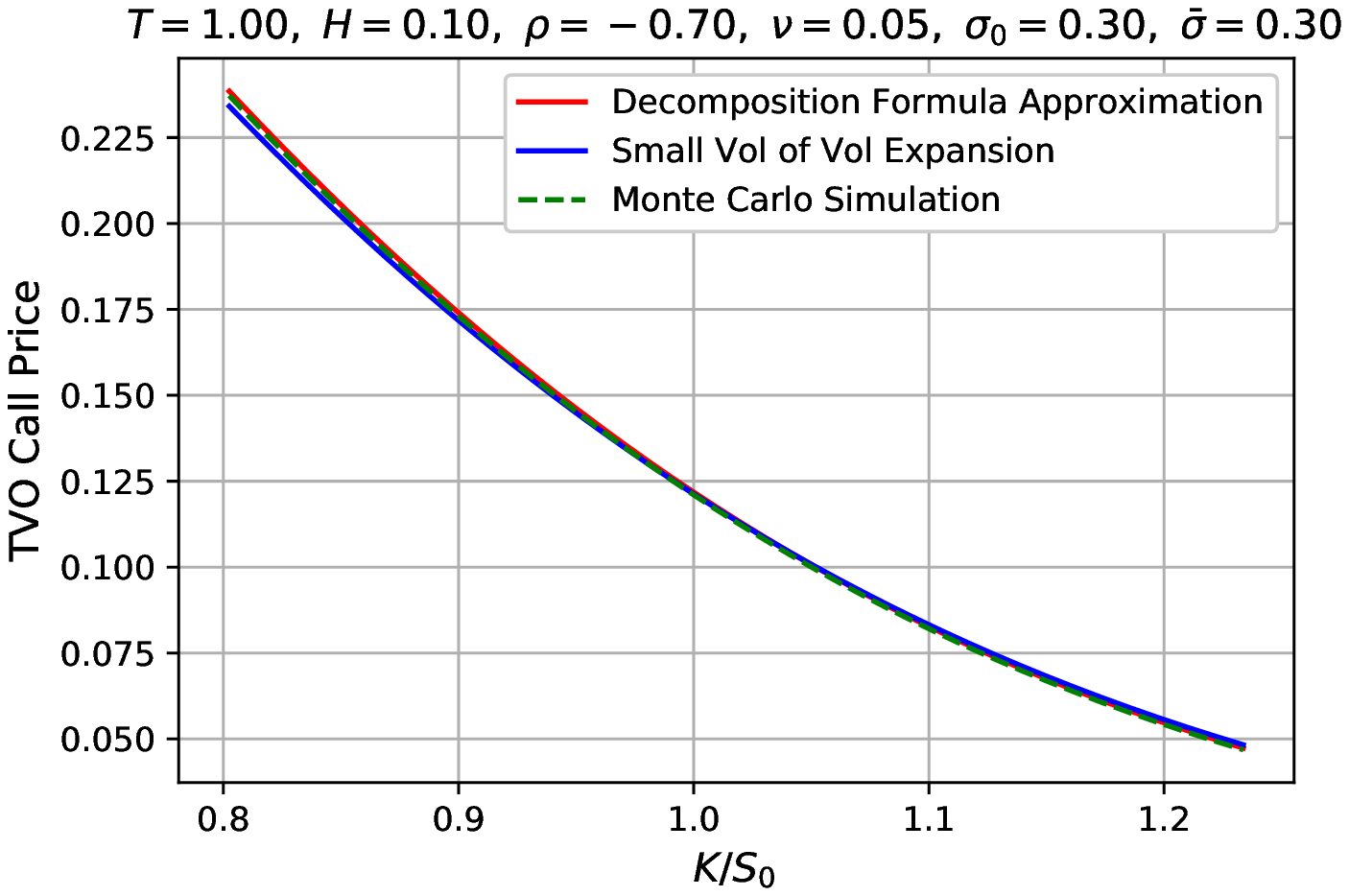} \quad 
\includegraphics[height=5cm]{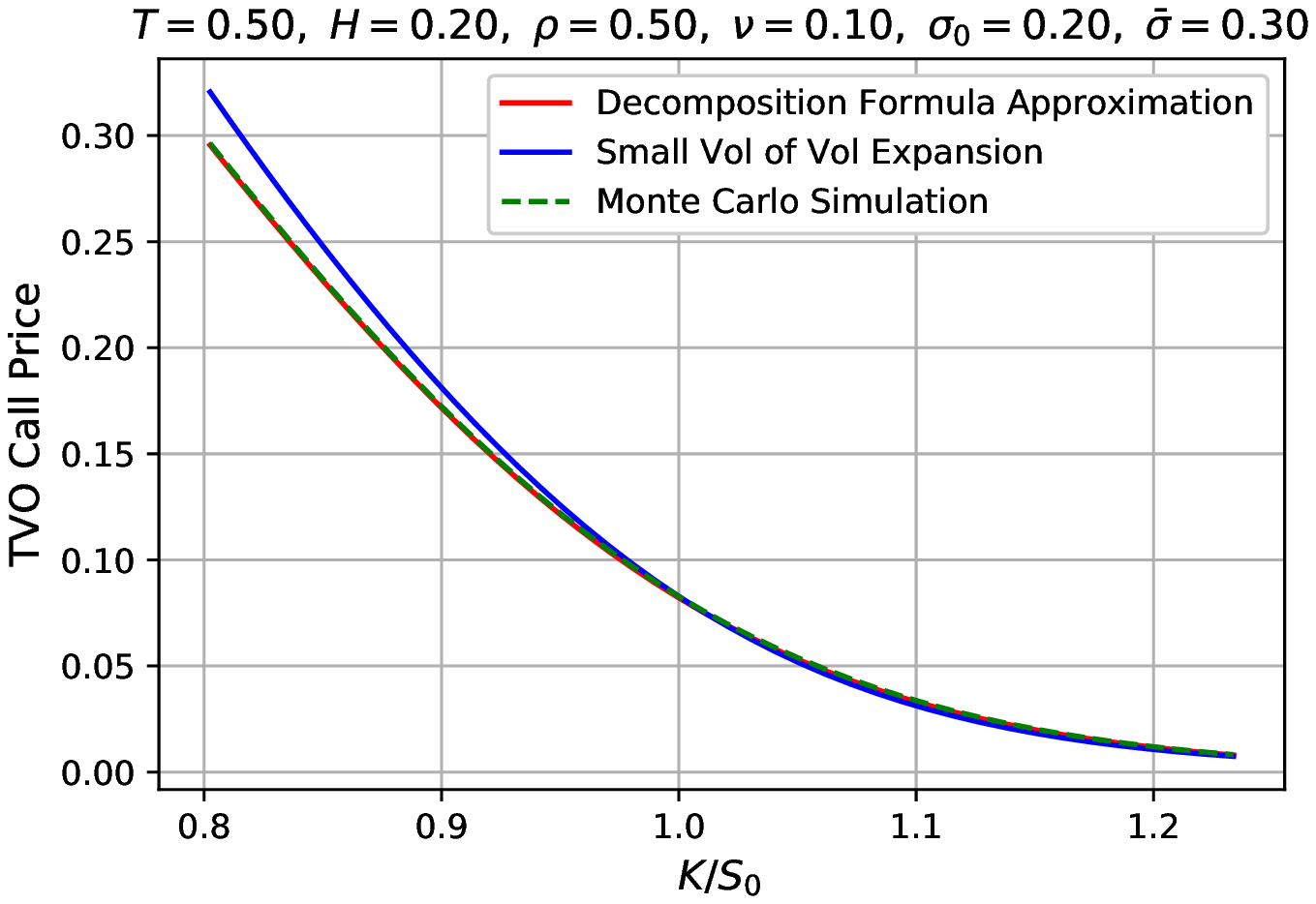}
\caption{TVO call price graphs for Table \ref{tab:t1} (left) and Table \ref{tab:t2} (right)}
\label{fig:f1}
\end{figure}

\begin{figure}
\centering
\includegraphics[height=5cm]{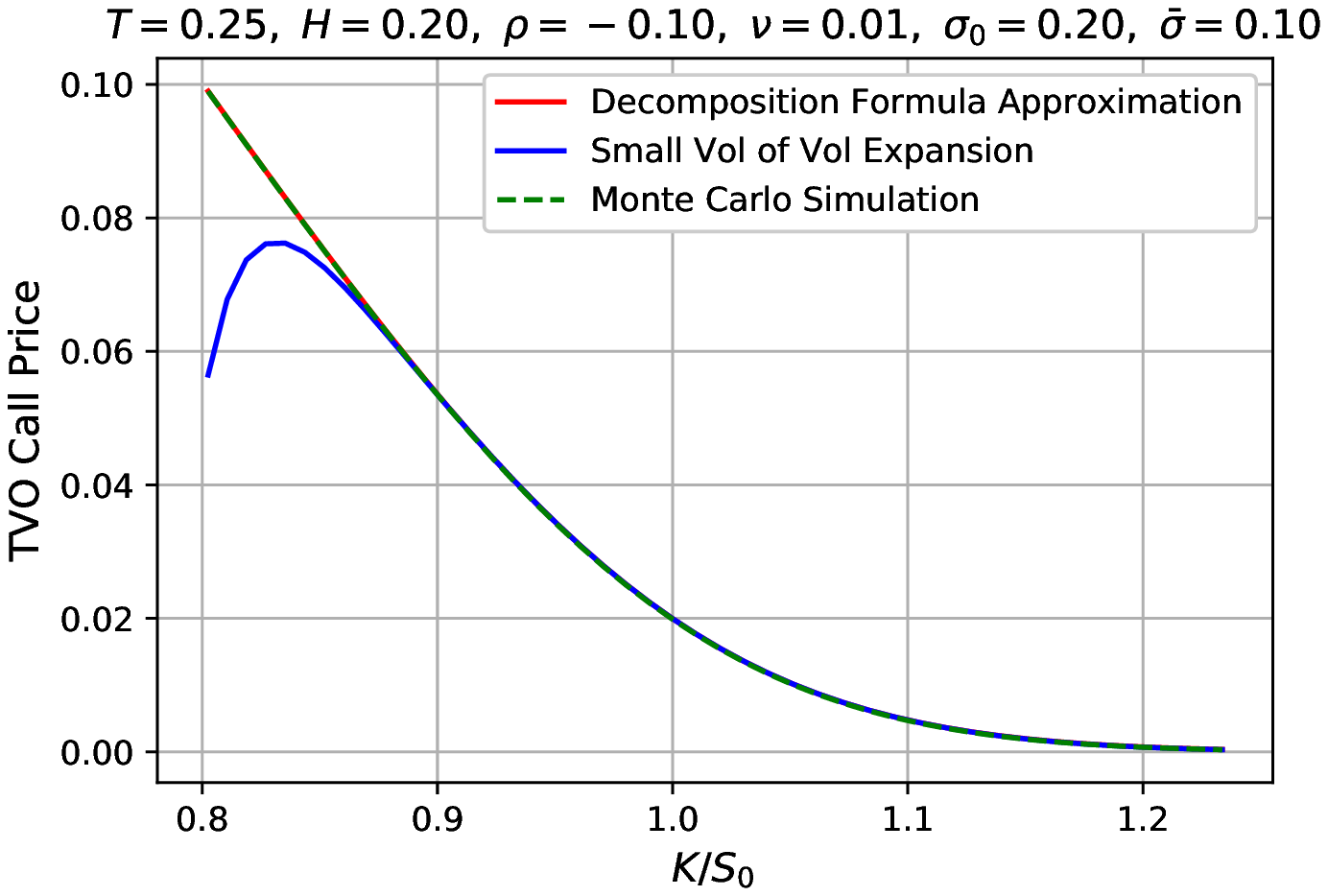} \quad
\includegraphics[height=5cm]{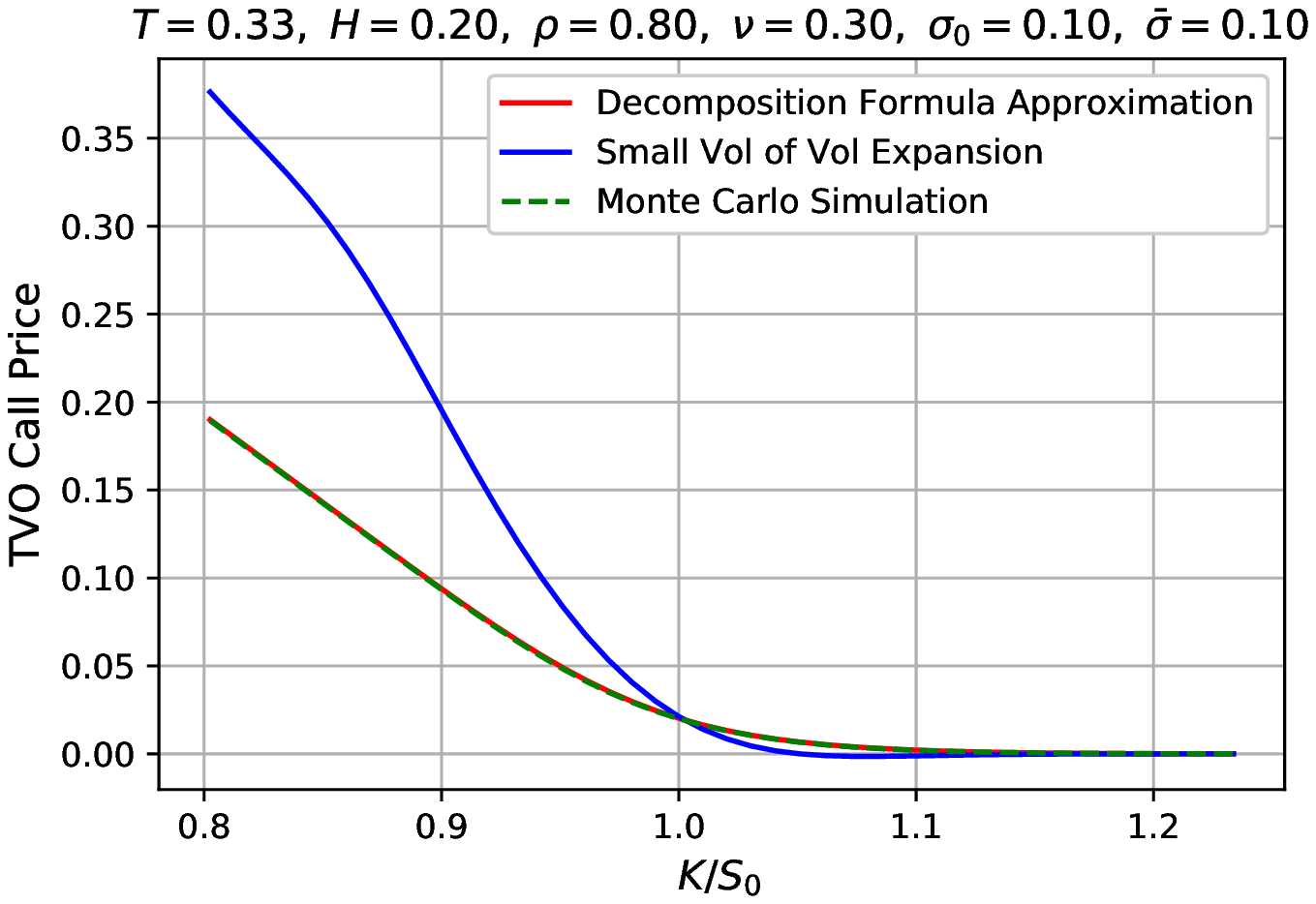}
\caption{TVO call price graphs for Table \ref{tab:t3} (left) and Table \ref{tab:t4} (right)}
\label{fig:f3}
\end{figure}

\begin{figure}
\centering
\includegraphics[height=5cm]{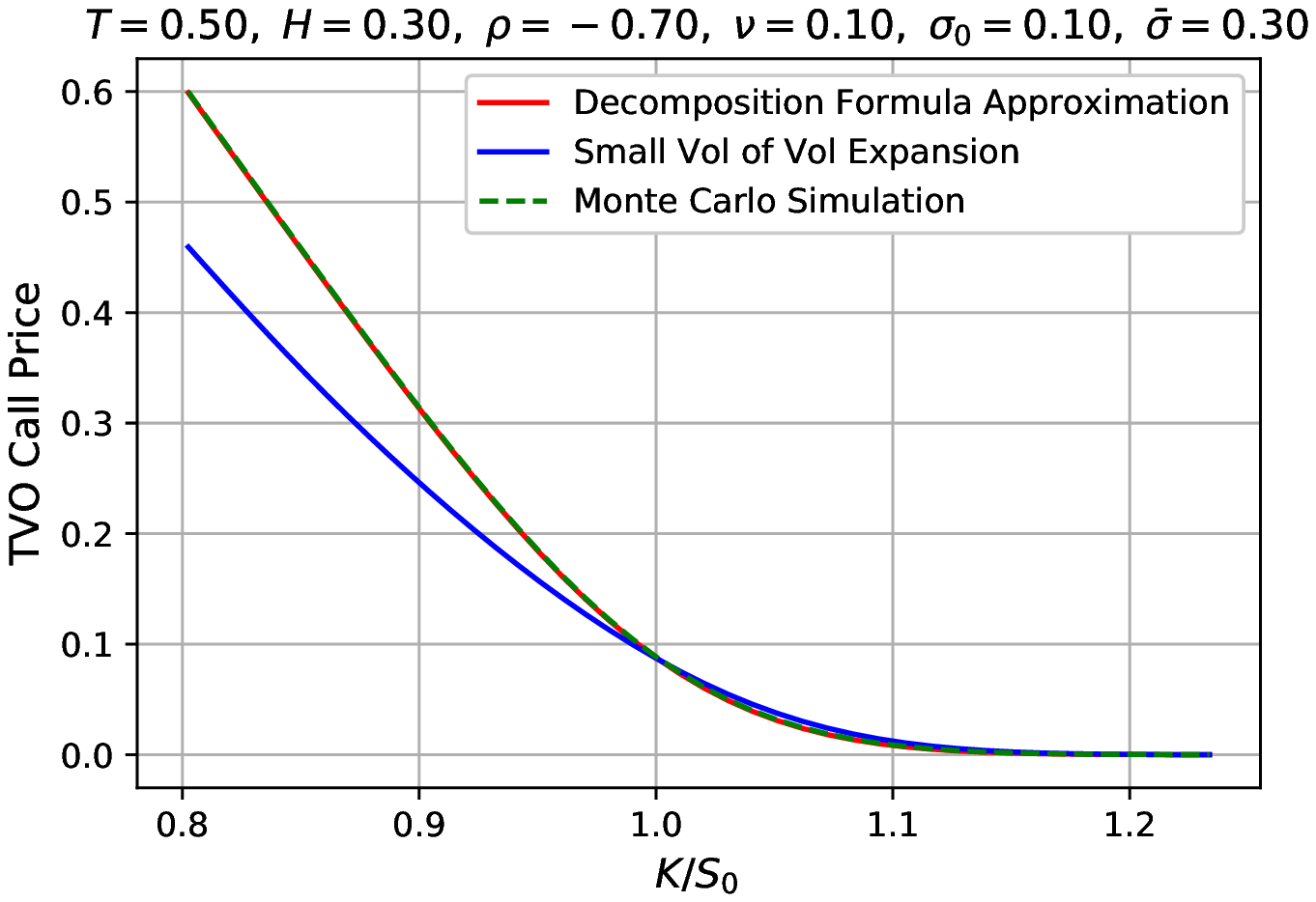} \quad 
\includegraphics[height=5cm]{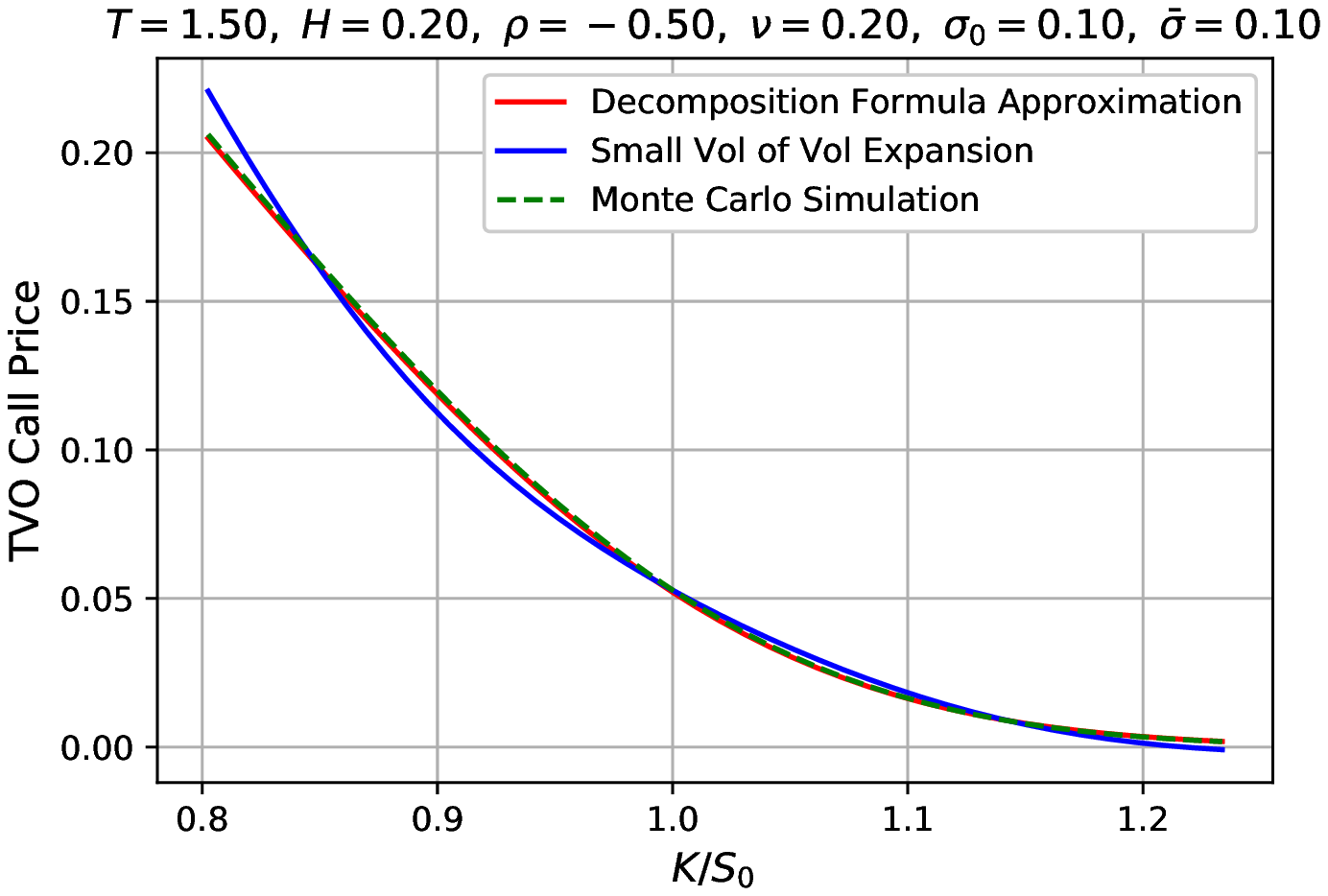}
\caption{TVO call price graphs for Table \ref{tab:t5} (left) and Table \ref{tab:t6} (right).}
\label{fig:f5}
\end{figure}

%%%%%%%%%%%%%%%%%SENSITIVITY%%%%%%%%%%%%%%%%%%%%%%%%%%%%%%%%%%%%%%%
\begin{figure}
\centering
\includegraphics[width=4.9cm]{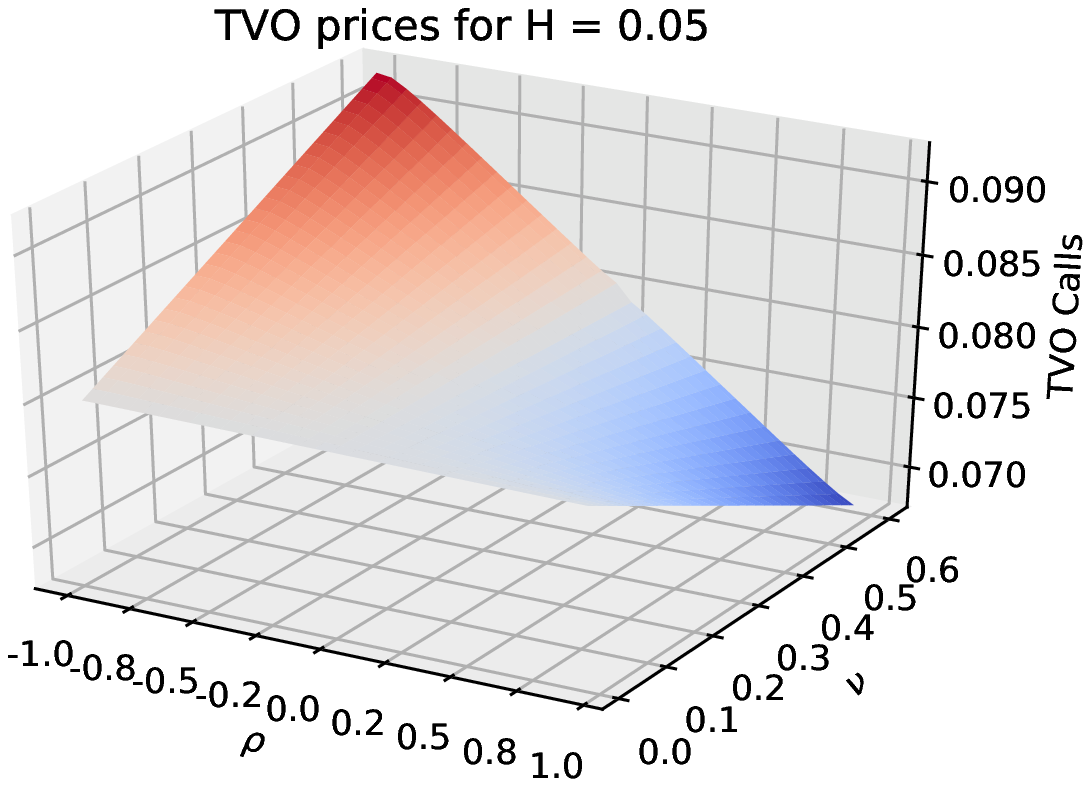}
\includegraphics[width=4.9cm]{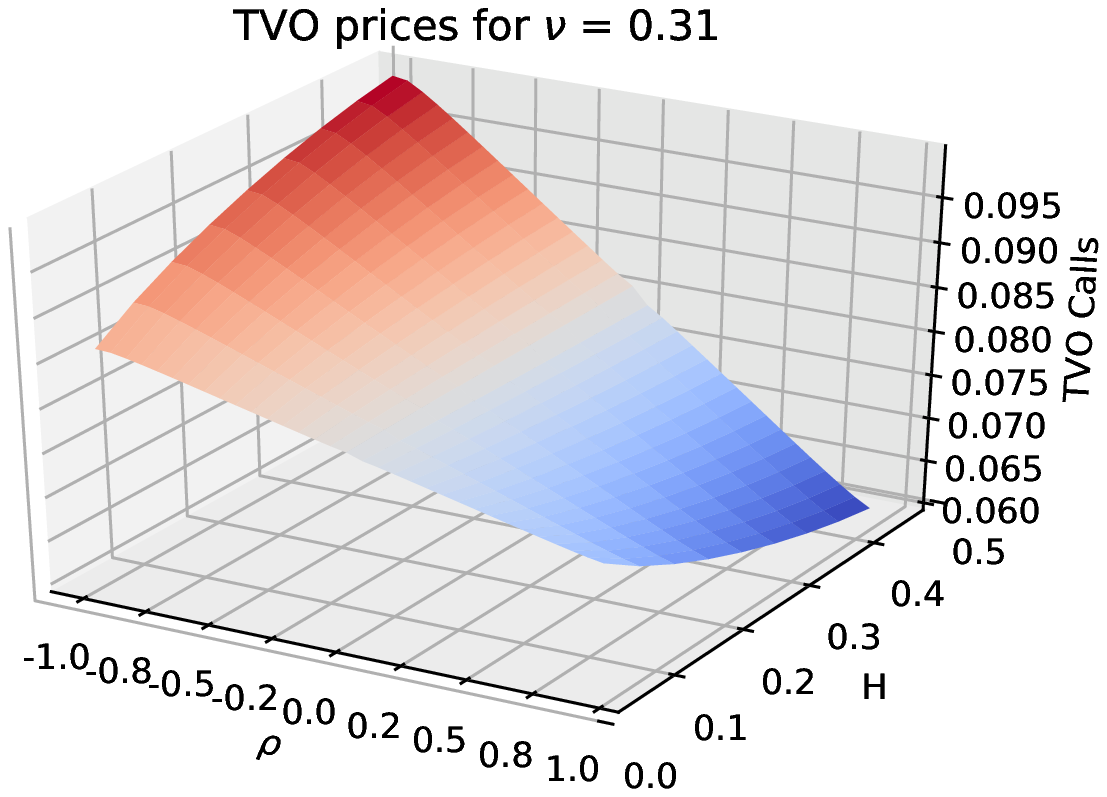}
\includegraphics[width=4.9cm]{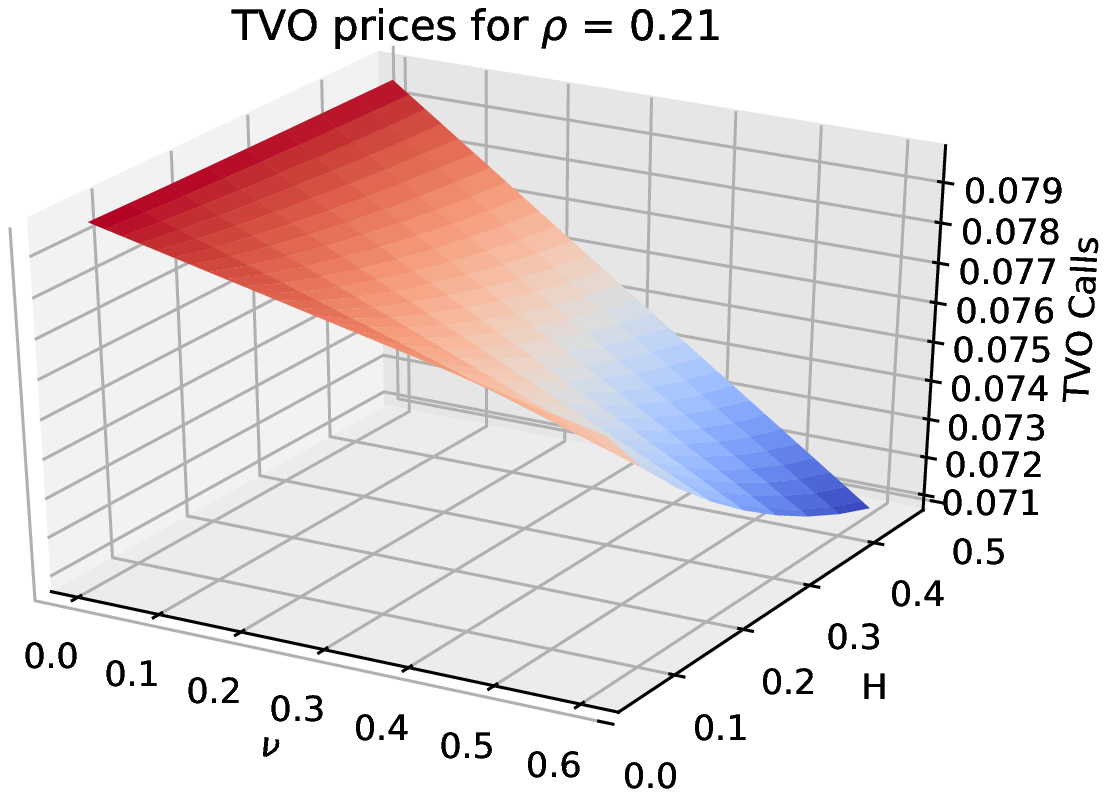}
\caption{Small Vol of Vol Expansion \eqref{eqn:tv-price-small-volvol} sensitivity to parameters. $\frac{K}{S_0}=1$, $T = 1, \sigma_0=0.1,\ \bar\sigma = 0.2$, $H\in(0,0.5),\ \nu\in(0,0.6),\ \rho\in(-1,1)$}
\label{fig:f7}
\end{figure}

\begin{figure}
\centering
\includegraphics[width=4.9cm]{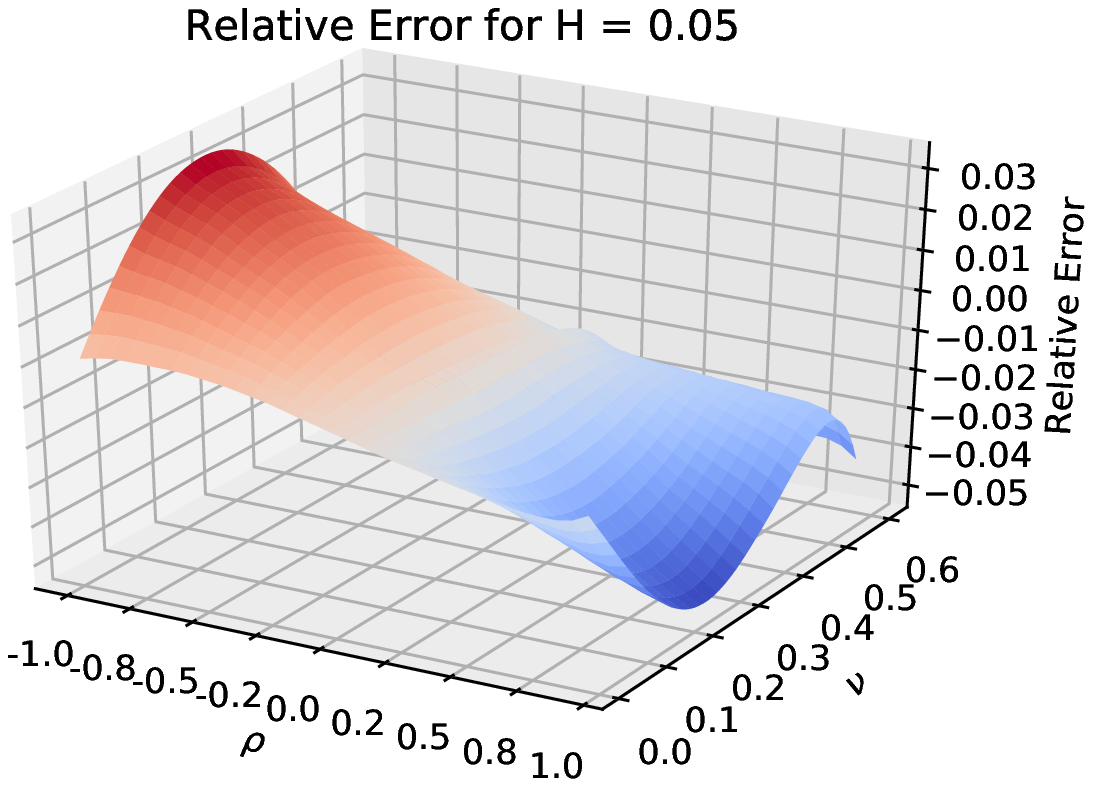}
\includegraphics[width=4.9cm]{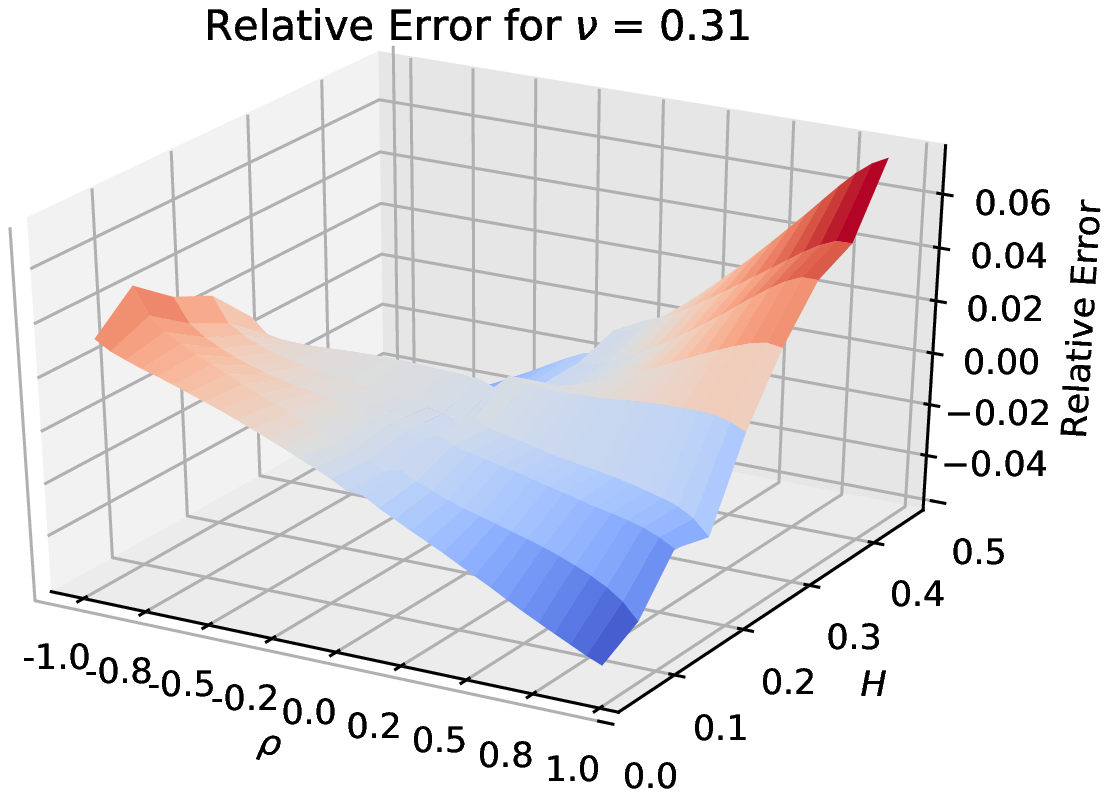}
\includegraphics[width=4.9cm]{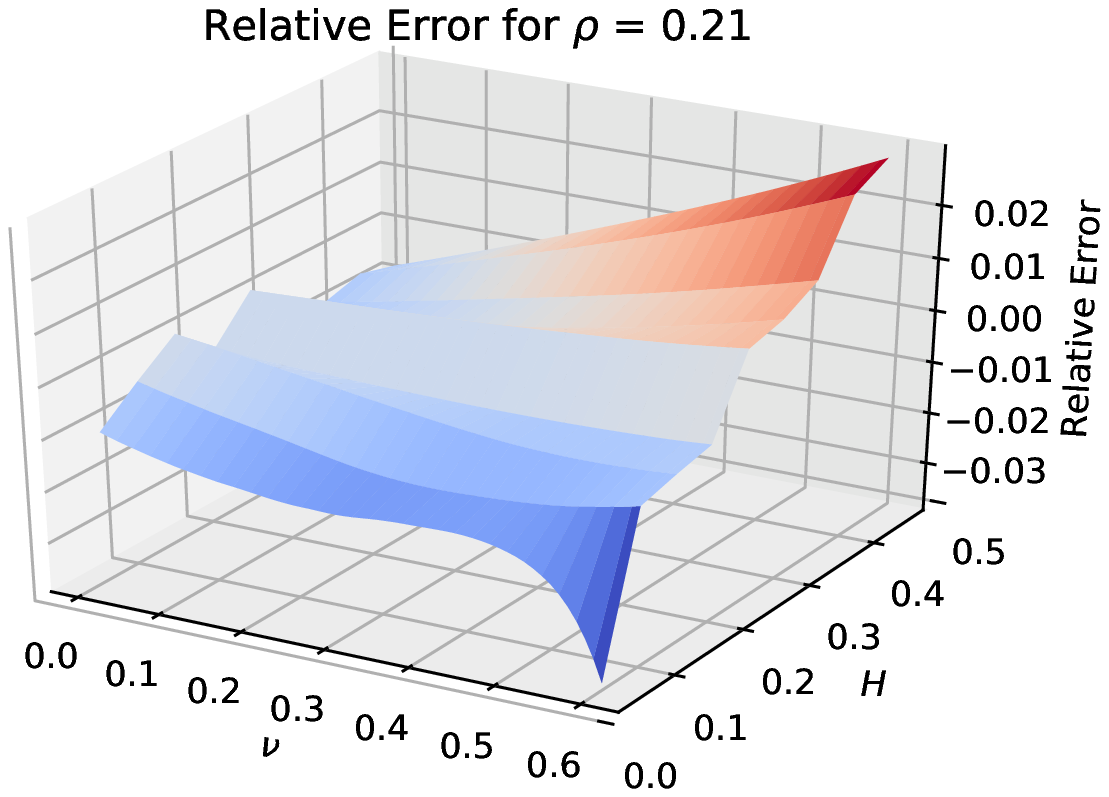}
\caption{Error surface between \eqref{eqn:tv-price-small-volvol} and MC. $\frac{K}{S_0}=1$, $T = 1, \sigma_0=0.1,\ \bar\sigma = 0.2$, $N=50,000$, $H\in(0,0.5),\ \nu\in(0,0.6),\ \rho\in(-1,1)$}
\label{fig:f8}
\end{figure}

%\begin{figure}
%\centering
%\includegraphics[scale = 0.85]{PriceFormula3D_fixH_2.eps}
%\includegraphics[scale = 0.85]{PriceFormula3D_fixNu_2.eps}
%\includegraphics[scale = 0.85]{PriceFormula3D_fixRho_2.eps}
%\caption{Small Vol of Vol Expansion \eqref{eqn:tv-price-small-volvol} sensitivity to parameters. We considered $S_0=K=1$, $T = 1$ year, $\sigma_0=0.1,\ \bar\sigma = 0.2$. The parameters are $H\in(0,0.5),\ \nu\in(0,0.6),\ \rho\in(-1,1)$}
%\label{fig:f9}
%\end{figure}

%\begin{figure}
%\centering
%\includegraphics[width=4.9cm]{errFormula3D_fixH_2.eps}
%\includegraphics[width=4.9cm]{errFormula3D_fixNu_2.eps}
%\includegraphics[width=4.9cm]{errFormula3D_fixRho_2.eps}
%\caption{Error surface between \eqref{eqn:tv-price-small-volvol} and MC price. $S_0=K=1$, $T = 1$ year, $\sigma_0=0.1,\ \bar\sigma = 0.2$, $N=50,000$. The parameters are $H\in(0,0.5),\ \nu\in(0,0.6),\ \rho\in(-1,1)$}
%\label{fig:f10}
%\end{figure}

%%%%%%%%%%%%%%%%%SENSITIVITY for DAF%%%%%%%%%%%%%%%%%%%%%%%%%%%%%%%%%%%%%%%
%\clearpage
\begin{figure}
\centering
\includegraphics[width=4.9cm]{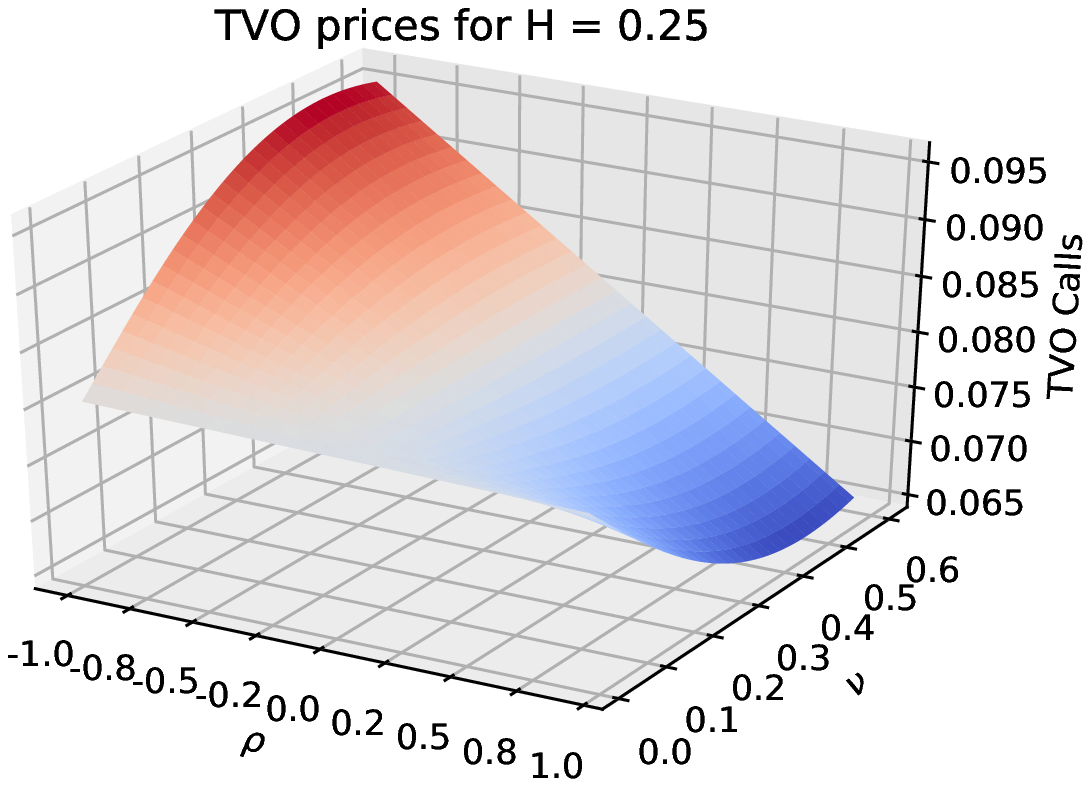}
\includegraphics[width=4.9cm]{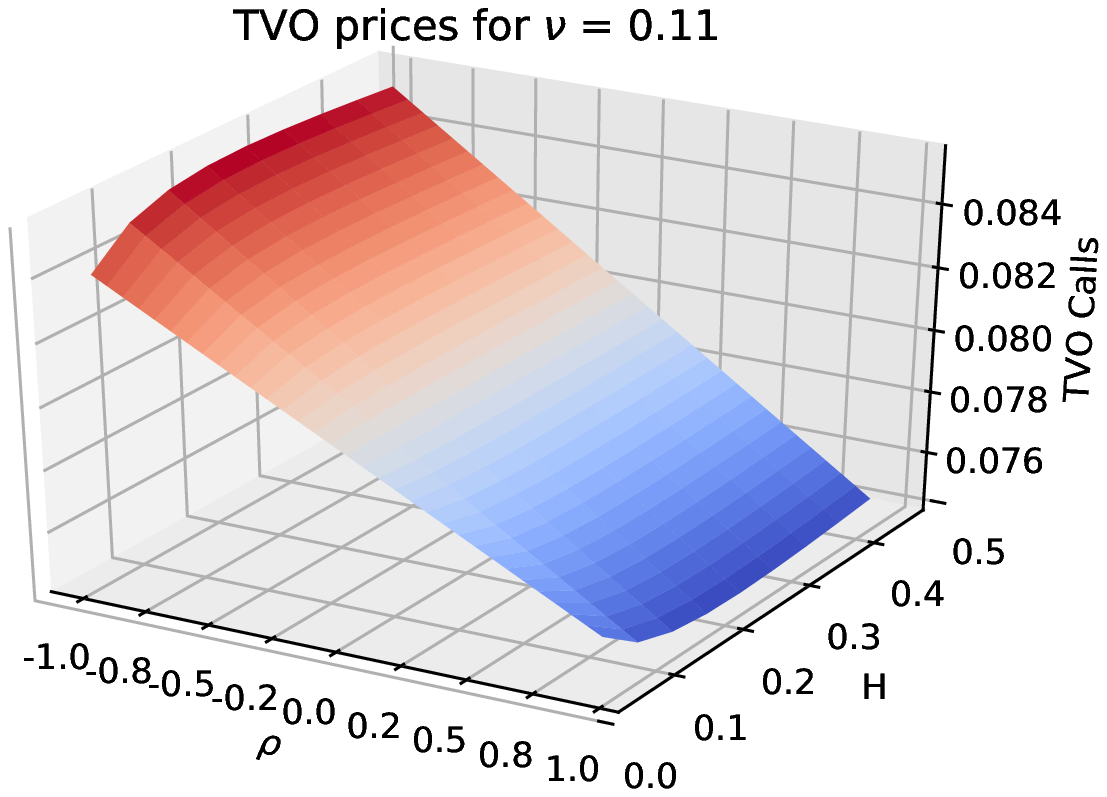}
\includegraphics[width=4.9cm]{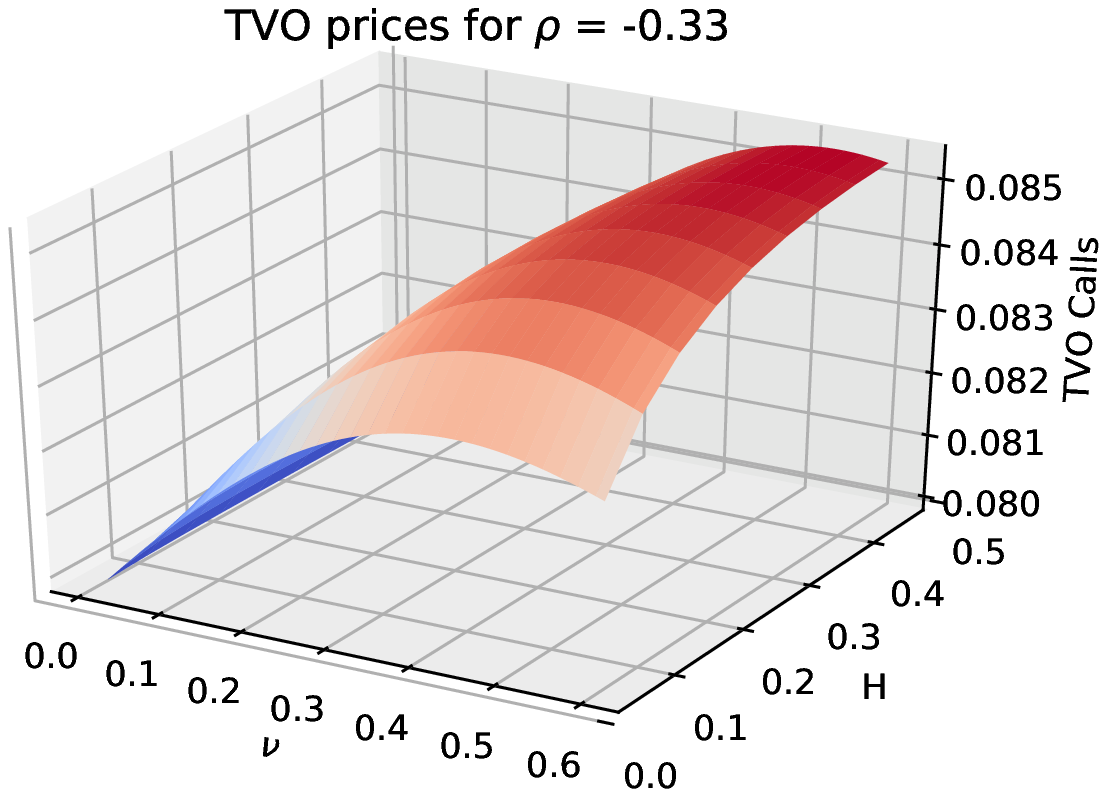}
\caption{Decomposition Formula Approximation \eqref{eqn:approx-tvc-fsabr-1-t0} sensitivity to parameters. $\frac{K}{S_0}=1, T = 1, \sigma_0=0.1,\bar\sigma = 0.2, H\in(0,0.5),\ \nu\in(0,0.6),\ \rho\in(-1,1)$}
\label{fig:f11}
\end{figure}

\begin{figure}
\centering
\includegraphics[width=4.9cm]{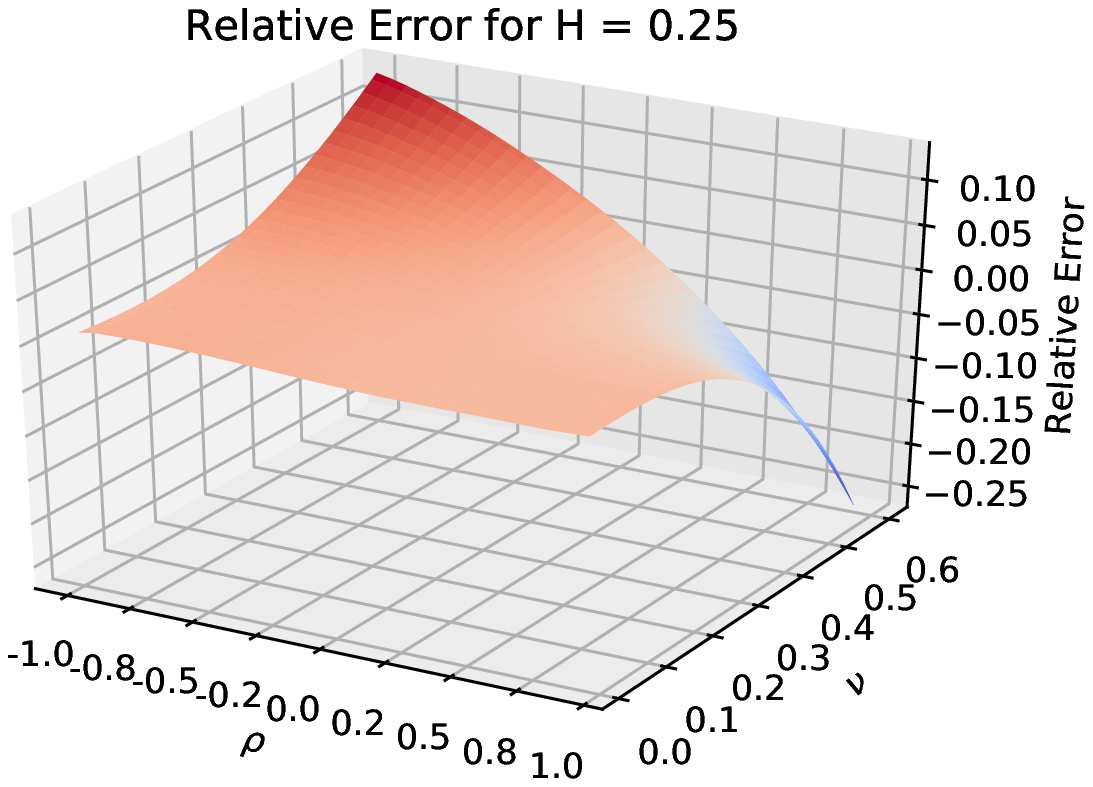}
\includegraphics[width=4.9cm]{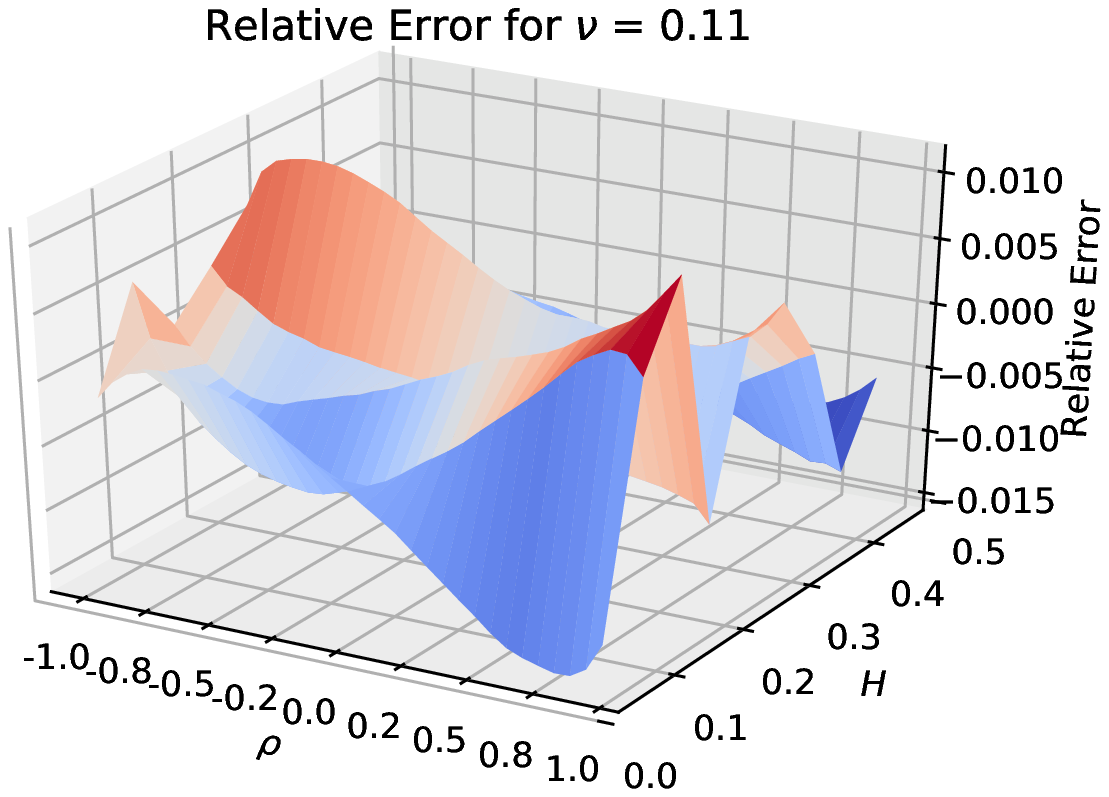}
\includegraphics[width=4.9cm]{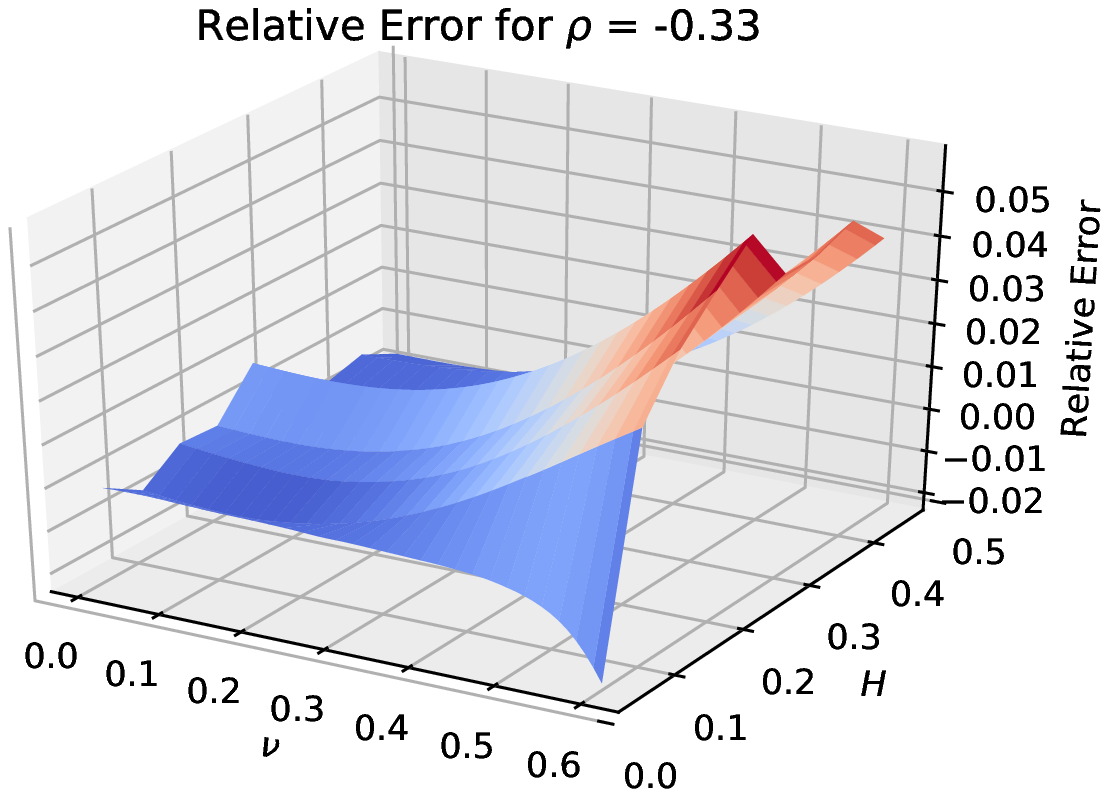}
\caption{Error surface between \eqref{eqn:approx-tvc-fsabr-1-t0} and MC. $\frac{K}{S_0}=1$, $T = 1,\sigma_0=0.1,\ \bar\sigma = 0.2$, $N=50,000$, $H\in(0,0.5),\ \nu\in(0,0.6),\ \rho\in(-1,1)$}
\label{fig:f12}
\end{figure}

%\begin{table}
%{\tiny
%\begin{adjustbox}{angle=90}
%\pgfplotstabletypeset[%
%   fixed zerofill,
%   precision=4,
%   col sep=space,
%   dec sep align,
%]{tableSensitivity_fixH.csv}
%\end{adjustbox}}
%\caption{}
%\label{tab:t5}
%\end{table}
%
%\begin{table}
%{\tiny
%\begin{adjustbox}{angle=90}
%\pgfplotstabletypeset[%
%   fixed zerofill,
%   precision=4,
%   col sep=space,
%   dec sep align,
%]{tableSensitivity_fixNu.csv}
%\end{adjustbox}}
%\caption{}
%\label{tab:t5}
%\end{table}
%
%\begin{table}
%{\tiny
%\begin{adjustbox}{angle=90}
%\pgfplotstabletypeset[%
%   fixed zerofill,
%   precision=4,
%   col sep=space,
%   dec sep align,
%]{tableSensitivity_fixRho.csv}
%\end{adjustbox}}
%\caption{}
%\label{tab:t5}
%\end{table}

%%%%%%%%%%%%%%%%%%%%%Implied Volatility%%%%%%%%%%%%%%%%%%%%%%%%%%%%%%%%%%
%\begin{figure}
%\centering
%\includegraphics[scale = 1]{IV_1.eps}
%\includegraphics[scale = 1]{IV_2.eps}
%\caption{Implied volatility from the standard Black--Scholes formula and Decomposition Formula Approximation \eqref{eqn:approx-tvc-fsabr-1-t0}}
%\label{fig:f13}
%\end{figure}
%
%\begin{figure}
%\centering
%\includegraphics[scale = 1]{IV_3.eps}
%\includegraphics[scale = 1]{IV_4.eps}
%\caption{Implied volatility from the standard Black--Scholes formula and Decomposition Formula Approximation \eqref{eqn:approx-tvc-fsabr-1-t0}}
%\label{fig:f14}
%\end{figure}

%%%%%%%%%%%%%%%%%%%%%%%%%%%%%%%%%%%%%%%%%%%%%%%%%%%%%%%%%%%%%%%%%
%  Section: Conclusion and discussion
%%%%%%%%%%%%%%%%%%%%%%%%%%%%%%%%%%%%%%%%%%%%%%%%%%%%%%%%%%%%%%%%%
\section{Conclusion and discussion} \label{sec:discussion}
Our aim of the paper was twofold. The first part derived the decomposition formulas in both It\^o and Malliavin calculus for the price of a target volatility call under a semiparametric model. The model considered here was {\it semiparametric} in the sense that it is a stochastic volatility model but without specifying explicitly the volatility process except certain technical conditions. In particular, the decomposition formula obtained by It\^o's calculus suggested a replicating strategy for target volatility option and an approximation formula for its price.  

In the second part of the paper, we specialized ourselves to the lognormal fractional SABR model that was recently suggested to the literature in stochastic volatility models because of its amazing fit to the empirical data of variance swaps. In other words, the volatility process was specified as the exponentiation of a scaled fractional Brownian motion. 
Explicit closed form approximation formulas in this model were derived from the decomposition formula and in the small volatility and volatility expansion. Numerical examples from Monte Carlo simulation showed that both approximation formulas worked well in a reasonable range of parameters. However, first order small volatility of volatility expansion broke down in extreme parameters as shown in the figures; whereas numerically approximation from decomposition formula passed the tests in a wider range of parameters. 

Efficient and accurate calculations or approximations of asset prices are crucial when it comes to calibrating the parameters to market data, especially when there is a process driven by fractional Brownian motion. The approximation formulas obtained in the current paper make this task easy due to their simplicity and accuracy. As market indicators, implied volatility from target volatility call options and possibly an implied Hurst exponent from the price of target volatility options can thus be defined accordingly. We leave all these discussions in a future research.

%%%%%%%%%%%%%%%%%%%%%%%%%%%%%%%%%%%%%%%%%%%%%%%%%%%%%%%%%%%%%%%%%
%  Acknowledgement
%%%%%%%%%%%%%%%%%%%%%%%%%%%%%%%%%%%%%%%%%%%%%%%%%%%%%%%%%%%%%%%%%

\section*{Acknowledgement} 
EA is partially supported by the Spanish grant MEC MTM2013-40782-P.
THW is partially supported by the Natural Science Foundation of China grant 11601018.

%%%%%%%%%%%%%%%%%%%%%%%%%%%%%%%%%%%%%%%%%%%%%%%%%%%%%%%%%%%%%%%%%
%  Section: Appendix
%%%%%%%%%%%%%%%%%%%%%%%%%%%%%%%%%%%%%%%%%%%%%%%%%%%%%%%%%%%%%%%%%

\section{Appendix - Technical results} \label{sec:appendix}
\begin{lemma} \label{lma:normal-expects}
Let $\xi$ be a normal random variable with mean $\mu$, variance $\sigma^2$ and $N(\cdot)$ denote the distribution function for standard normal. Then we have
\bea
\Eof{N(\xi)} &=& N\left( \frac\mu{\sqrt{1 + \sigma^2}}\right) \label{eqn:E-NX} \\
%&& \Eof{e^\xi N(\xi)} = e^{\mu + \frac{\sigma^2}2} N\left(\frac\mu{\sqrt{1 + \sigma^2}} - \rho\sigma \right). \label{eqn:E-eXNX} 
 \Eof{\xi N(\xi)} &=& \mu N\left(\frac\mu{\sqrt{1 + \sigma^2}} \right) + \frac{\sigma^2}{\sqrt{1 + \sigma^2}} N'\left(\frac\mu{\sqrt{1 + \sigma^2}} \right) \label{eqn:E-XNX} \\
\Eof{e^{a\xi} N(\xi)} &=& e^{a\mu + \frac{a^2\sigma^2}2} N\left(\frac{\mu + a \sigma^2}{\sqrt{1 + \sigma^2}} \right) \label{eqn:E-eXNX} \\
%&& \Eof{e^{a\xi^2 + b\xi} N(\xi)} = e^{a\mu + \frac{a^2\sigma^2}2} N\left(\frac{\mu + a \sigma^2}{\sqrt{1 + \sigma^2}} \right), \label{eqn:E-eXNX-2}
 \Eof{\xi e^{a\xi} N(\xi)} &=&e^{a\mu + \frac{a^2\sigma^2}2}\times\nonumber\\
&&\left[ \left(\mu + \frac{a\sigma^2}2\right) N\left(\frac{\mu + a \sigma^2}{\sqrt{1 + \sigma^2}} \right) + \frac{\sigma^2}{\sqrt{1 + \sigma^2}} N'\left(\frac{\mu + a \sigma^2}{\sqrt{1 + \sigma^2}} \right) \right] \label{eqn:E-eXNX-1} \\
 \Eof{\xi^2 e^{a\xi} N(\xi)} &=&  e^{a\mu + \frac{a^2\sigma^2}2}\Big[ \left( (\mu+a\sigma^2)(\mu+\frac{1}{2}a\sigma^2)+\frac{\sigma^2}{2}\right)N\left(\frac{\mu + a \sigma^2}{\sqrt{1 + \sigma^2}} \right)+\nonumber\\
 &&\left(2\mu +\frac{3}{2}a\sigma^2\right)\frac{\sigma^2}{\sqrt{1+\sigma^2}}N'\left(\frac{\mu + a \sigma^2}{\sqrt{1 + \sigma^2}} \right)+\nonumber\\
 &&\frac{\sigma^4}{1+\sigma^2}N''\left(\frac{\mu + a \sigma^2}{\sqrt{1 + \sigma^2}} \right)\Big]\label{eqn:E-eXNX-2}
\eea
for any constant $a$.
\end{lemma}
\begin{proof}
%Straightforward calculations. 
We prove only \eqref{eqn:E-eXNX} since \eqref{eqn:E-NX} is readily obtained by setting $a=0$ and \eqref{eqn:E-eXNX-1}, \eqref{eqn:E-eXNX-2} can be obtained by differentiating \eqref{eqn:E-eXNX} with respect to $a$.
Consider 
\beaa
&& \Eof{e^{a\xi} N(\xi)} = \Eof{e^{a\xi} \Pof{Z \leq \xi | \xi} } = \Eof{e^{a\xi} \1_{\{Z \leq \xi\}}} = \Eof{e^{a\xi} \1_{\{Y \leq 0\}}} 
\eeaa
where $Y = Z - \xi$. Note that we can decompose $\xi$ as 
\[
\xi = \mu + \frac{\cov(\xi,Y)}{\var(Y)} (Y - \Eof{Y}) + \sqrt{1 - \rho^2} \sigma B
\]
where $B$ is standard normal, independent of $Y$ and $\rho$ is the correlation between $\xi$ and $Y$. Indeed, in our case
\[
\xi = \mu - \frac{\sigma^2}{1 + \sigma^2} (Y + \mu) + \frac\sigma{\sqrt{1 + \sigma^2}} B.
\]
It follows that, since $Y$ and $B$ are independent, 
\beaa
&& \Eof{e^{a\xi} \1_{\{Y \leq 0\}}} = \Eof{e^{a\left(\mu - \frac{\sigma^2}{1 + \sigma^2} (Y + \mu) + \frac\sigma{\sqrt{1 + \sigma^2}} B \right)} \1_{\{Y \leq 0\}}} \\
&=& e^{a\mu}\, \Eof{e^{- \frac{a \sigma^2}{1 + \sigma^2} (Y + \mu)} \1_{\{Y \leq 0\}}} \,
\Eof{e^{\frac{a\sigma}{\sqrt{1 + \sigma^2}} B}}.
\eeaa
Finally, by straightforward calculations, one can show that the last expression is indeed equal to the right hand side of \eqref{eqn:E-eXNX}. 
\end{proof}
Denote by $h_n(\cdot)$ the normalized Hermite polynomials, i.e., for $n \geq 0$, 
\begin{equation}
h_n(x) = \frac{(-1)^n}{\sqrt{n!}} e^{\frac{x^2}2}\frac{d^n}{dx^n}(e^{-\frac{x^2}2}). \label{eqn:hermite-poly}
\end{equation}  
Note that since 
\beaa
&& \Eof{h_n(Z) h_m(Z)} = \delta_{nm} 
\eeaa
for a standard normally distributed random variable $Z$, the set $\{h_n(Z)\}_{n=0}^\infty$ forms an orthonormal basis for the $\sigma$-algebra generated by $Z$. 
\begin{lemma} \label{lma:cond-exps}
Let 
\beaa
&& \mu_t = \frac{c_H B_T}T t^{H + \frac12}, \qquad 
\sigma^2_t = t^{2H} - \frac{c_H^2}T t^{2H + 1}.
\eeaa
Then, for $k \geq 1$, 
\bea
&& \Eof{\left. (B_t^H)^k \right|B_T} = M_k(\mu_t, \sigma^2_t), \label{eqn:Bh-BT} \\
&& \Eof{\left. \int_0^T (B_t^H)^k dB_t \right|B_T} = \frac{c_H^k \sqrt{(k+1)!}}{k \left(H + \frac12 \right) + 1} \; T^{kH + \frac12} h_{k+1}\left( \frac{B_T}{\sqrt T} \right), \label{eqn:intBh-BT}
\eea
where $M_k(\mu, \sigma^2)$ is the $k$th moment of a normal random variable with mean $\mu$ and variance $\sigma^2$. $c_H = $ is a constant. In particular, for $k=1,2$, we have 
\beaa
&& \Eof{\left. \int_0^T B_t^H dt \right|B_T} = \frac{2c_H}{2H + 3} T^{H + \frac12} B_T, \\
&& \Eof{\left. \int_0^T (B_t^H)^2 dt \right|B_T} = \left(\frac1{2H+1} - \frac{c_H^2}{2H+2}\right) T^{2H + 1}, \\
&& \Eof{\left. \int_0^T B_t^H dB_t \right|B_T} = \frac{2c_H}{2H + 3} T^{H + \frac12} \left(B_T^2 - T\right), \\
&& \Eof{\left. \int_0^T (B_t^H)^2 dB_t \right|B_T} = \frac{c_H^2}{2(H + 1)} T^{2H + \frac12} \left\{\left(\frac{B_T}{\sqrt T}\right)^3 - 3 \frac{B_T}{\sqrt T} \right\}. 
\eeaa
\end{lemma}
\begin{proof}
\eqref{eqn:Bh-BT} follows from straightforward calculations since, conditioned on $B_T$, $B_t^H$ is normally distributed with mean $\mu_t$ and variance $\sigma_t$. 
%\[
%B_t^H | \cF^{B_T} \sim N\left(\frac{c_H B_T}T t^{H + \frac12}, t^{2H} - \frac{c_H^2}T t^{2H + 1}\right).
%\]
For \eqref{eqn:intBh-BT}, notice that the $\sigma$-algebra $\cF^{B_T}$ is spanned by $\left\{h_n\left( \frac{B_T}{\sqrt T} \right)\right\}_{n=0}^\infty$, where the $h_n$'s are the Hermite polynomials defined in \eqref{eqn:hermite-poly}. Consider 
\[
\int_0^T (B_t^H)^k dB_t = \sum_{n=0}^\infty c_n h_n\left( \frac{B_T}{\sqrt T} \right) + \xi,
\]
where $\xi$ has mean zero and is orthogonal to the span of $\left\{h_n\left( \frac{B_T}{\sqrt T} \right) \right\}_{n=0}^\infty$. Since the random variables $\left\{h_n\left( \frac{B_T}{\sqrt T} \right)\right\}_{n=0}^\infty$ form an orthonormal basis, it follows that
\[
c_n = \Eof{\int_0^T B_t^H dB_t \; h_n\left( \frac{B_T}{\sqrt T} \right)}.
\] 
For $k \geq 1$, denote by $\bs = (s_1, \cdots, s_k)$, $d\bs = ds_1 \cdots ds_k$, $dB_{\bs} = dB_{s_1} \cdots dB_{s_k}$, and $\Delta_k = \{\bs : 0 \leq s_1 \leq \cdots \leq s_k \leq t\}$ hereafter for notational simplicity.  
Notice that since $h_n\left( \frac{B_T}{\sqrt T} \right)$ can be written as an $n$-iterated Wiener integral of constant function $1$ as
\[
h_n\left( \frac{B_T}{\sqrt T} \right) = \sqrt{\frac{n!}{T^n}} \int_{\Delta_n} dB_{\bs},
\] 
one can easily verify that $c_n = 0$ for all $n \neq k+1$. As for $n=k+1$, we have
\beaa
c_{k+1} &=& \Eof{\int_0^T (B_t^H)^k dB_t \; h_{k+1}\left(\frac{B_T}{\sqrt T} \right)} \\
&=& \sqrt{\frac{(k+1)!}{T^{k+1}}} \; k! \; \Eof{\int_0^T\int_{\Delta_k} K(t,\bs)dB_{\bs} dB_t \int_0^T\int_{\Delta_k} dB_{\bs} dB_t} \\
&=& \sqrt{\frac{(k+1)!}{T^{k+1}}} \; k! \; \int_0^T \int_{\Delta_k} K(t,\bs) d\bs dt \\
&=& \sqrt{\frac{(k+1)!}{T^{k+1}}} \; \int_0^T \left(\int_0^t K(t,s) ds \right)^k dt \\
&=& c_H^k \sqrt{\frac{(k+1)!}{T^{k+1}}} \int_0^T t^{k\left(H + \frac12\right)} dt \\
&=& c_H^k \sqrt{\frac{(k+1)!}{T^{k+1}}} \frac{T^{k\left(H + \frac12\right) + 1}}{k\left(H + \frac12\right) + 1} \\
&=& \frac{c_H^k \sqrt{(k+1)!}}{k \left(H + \frac12 \right) + 1} \; T^{kH + \frac12},
\eeaa 
where in the third equality we used the property that, for any given deterministic function $f(\bs)$ defined on $\Delta_k$, 
\[
\Eof{\int_{\Delta_k}f(\bs) dB_{\bs} \int_{\Delta_k} 1 dB_{\bs}} = \int_{\Delta_k} f(\bs) d\bs
\]
obtained by iteratively using It\^o's isometry. We conclude that 
\beaa
\Eof{\left. \int_0^T (B_t^H)^k dB_t \right|B_T} = c_{k+1} \Eof{\left. h_{k+1}\left( \frac{B_T}{\sqrt T} \right) \right| B_T} = \frac{c_H^k \sqrt{(k+1)!}}{k \left(H + \frac12 \right) + 1} \; T^{kH + \frac12} h_{k+1}\left( \frac{B_T}{\sqrt T} \right).
\eeaa
\end{proof}

\begin{lemma} \label{lma:cond-fbm}
For $t < r$, conditioned on $\cF_t$, $B^H_r$ is normally distributed with mean $m(r|t)$ and variance $v(r|t)$ given respectively by 
\[
m(r|t) = \int_0^t K(r,s) dB_s, \quad v(r|t) = \int_t^r K^2(r,s)ds.
\]
\end{lemma}
\begin{proof}
Consider the characteristic function of $B_r^H$ conditioned on $\cF_t$ 
\beaa
&& \Et{e^{iuB^H_r}} = e^{iu\int_0^t K(r,s) dB_s}\Et{e^{iu\int_t^r K(r,s)dB_s}} \\
&=& e^{iu\int_0^t K(r,s) dB_s}\E\left[e^{iu\int_t^r K(r,s)dB_s}\right] \\
&=& e^{iu\int_0^t K(r,s) dB_s - \frac{u^2}2 \int_t^r K^2(r,s) ds}. 
\eeaa
It follows that, conditioned on $\cF_t$, $B_r^H$ is normally distributed with mean and variance given by $m(r|t)$ and $v(r|t)$ respectively.
\end{proof}

%%%%%%%%%%%%%%%%%%%%%%%%%%%%%%%%%%%%%%%%%%%%%%%%%%%%%%%%%%%%%%%%%%%%%%%%%%%%%%%%%%%%%%%%
%  Proof of computational formula Corollary
%%%%%%%%%%%%%%%%%%%%%%%%%%%%%%%%%%%%%%%%%%%%%%%%%%%%%%%%%%%%%%%%%%%%%%%%%%%%%%%%%%%%%%%%

\subsection{Proof of Corollary \ref{cor:tvoComputationFormula}}
\label{sec:proof-of-cor}
We calculate each individual expectation in Theorem \ref{thm:tv-price-small-volvol} in the following. 
For the first order term, we calculate the two terms separately:
\beaa
&& \Eof{C_w w_T^{(1)}} = \Eof{C_w Y_0^2 \, 2 \int_0^T B_t^H dt} 
= 2 Y_0^2 \, \Eof{C_w \E\left[\left. \int_0^T B_t^H dt \right| B_T \right]} \\
&=& 2 Y_0^2 \, \frac{2c_H}{2H + 3} T^{H + \frac12} \, \Eof{C_w \, B_T},
\eeaa
where in passing to the second equality we used the fact that $C_w$ is $\cF^{B_T}$-measurable and in passing to the last equality we used \eqref{eqn:Bh-BT} with $k=1$.
For the other term,
\beaa
&& \Eof{C_x \, \xi_T^{(1)}} 
= \Eof{C_x \, \left\{\rho Y_0 \int_0^T B_t^H dB_t - \frac{\rho^2}2 w_T^{(1)} \right\}} \\
&=& \rho Y_0 \Eof{C_x \, \E\left[\left. \int_0^T B_t^H dB_t \right|B_T \right]} 
- \rho^2 Y_0^2 \Eof{C_x \, \E\left[\left. \int_0^T B_t^H dt \right|B_T \right]} \\
%&=& \frac{2c_H \rho Y_0}{2H + 3} T^{H + \frac12} \Eof{C_x\left(\xi_T^{(0)}, \bar\rho^2 Y_0^2 T \right) (B_T^2 - T)} \\
%&& - \frac{2 c_H \rho^2 Y_0^2}{2H+3} T^{H + \frac12} \Eof{C_x\left(\xi_T^{(0)}, \bar\rho^2 Y_0^2 T \right) B_T } \\
&=& \frac{2c_H \rho Y_0}{2H + 3} T^{H + \frac12} \Eof{C_x \, \left\{ B_T^2 - T - \rho Y_0 B_T \right\}},
\eeaa
where in passing to the second equality we used the fact that $C_x$ is $\cF^{B_T}$-measurable and in passing to the last equality we used \eqref{eqn:Bh-BT} and \eqref{eqn:intBh-BT} with $k=1$. Furthermore, since $C$ satisfies $C_w = \frac12(C_{xx} - C_x)$, the first order term becomes
\beaa
&& \Eof{C_x \xi_T^{(1)} + \bar\rho^2 C_w w_T^{(1)}} \\
&=& \frac{2c_H}{2H + 3} T^{H + \frac12} \Eof{2 \bar\rho^2 Y_0^2 C_w \, B_T + \rho Y_0 C_x \, \left\{ B_T^2 - T - \rho Y_0 B_T \right\}} \\
&=& \frac{2c_H}{2H + 3} T^{H + \frac12} \Eof{\bar\rho^2 Y_0^2 C_{xx} \, B_T - Y_0^2 C_x \, B_T + \rho Y_0 C_x \, \left\{ B_T^2 - T\right\}}.
\eeaa
By a similar argument, we can calculate:
\[
\Eof{w_T^{(1)}C} = 2 Y_0^2 \, \frac{2c_H}{2H + 3} T^{H + \frac12} \, \left( \Eof{B_T e^{\xi_T^{(0)}} N(d_1)}-\Eof{B_T N(d_2)}\right)
\]
Hence, the price of a TV call to the first order of $\nu$ is:
\beaa
\text{TVO} &=& \frac{\bar\sigma K}{Y_0}C(X_0, Y_0^2T) + \sum_{j=1}^5\kappa_j \mathcal{E}_j + \mathcal{O}(\nu^2),
\eeaa
where
\beaa
\kappa_1 &=& \frac{\sqrt{\frac{2}{\pi }} c_H K \nu  \sqrt{1-\rho ^2} \bar\sigma T^H}{2 H+3}\\
\kappa_2 &=& \frac{2 c_H K \nu  \bar\sigma T^{H-\frac{1}{2}}}{(2 H+3) Y_0}\\
\kappa_3 &=& -\frac{2 c_H K \nu  \bar\sigma T^{H-\frac{1}{2}} \left(\rho ^2 Y_0^2 T+1\right)}{(2 H+3) Y_0}\\
\kappa_4 &=& \frac{2 c_H K \nu  \rho  \bar\sigma T^{H+\frac{1}{2}}}{2 H+3}\\
\kappa_5 &=& -\frac{2 c_H K \nu  \rho  \bar\sigma T^{H+\frac{3}{2}}}{2 H+3}\\
\mathcal{E}_1 &=& \Eof{B_Te^{-\frac{d_2^2}{2}}}\\
\mathcal{E}_2 &=& \Eof{B_T N(d_2)}\\
\mathcal{E}_3 &=& \Eof{B_T e^{\xi_T^{(0)}} N(d_1)}\\
\mathcal{E}_4 &=& \Eof{B_T^2 e^{\xi_T^{(0)}} N(d_1)}\\
\mathcal{E}_5 &=& e^{X_0}N(d_1),
\eeaa
and $$d_{1,2} = \frac{X_0}{\sqrt{Y_0^2T}}\pm\frac{1}{2}\sqrt{Y_0^2T}.$$ Finally, by applying the identities in Lemma \ref{lma:normal-expects} and straightforward but tedious calculations, we obtain the small volatility of volatility expansion up to the first order. Note that the formula can be easily implemented numerically.

%%%%%%%%%%%%%%%%%%%%%%%%%%%%%%%%%%%%%%%%%%%%%%%%%%%%%%%%%%%%%%%%%%%%%%%%%%%%%%%%%%%%%%%%
%  Proof of 
%%%%%%%%%%%%%%%%%%%%%%%%%%%%%%%%%%%%%%%%%%%%%%%%%%%%%%%%%%%%%%%%%%%%%%%%%%%%%%%%%%%%%%%%
\subsection{Proof of Lemma \ref{lma:cond-exps}}
Recall that 
\[ 
M_t = \int_0^T \Et{Y_r^2} dr
\]
is a martingale. By applying the Clark-Ocone formula (assume we can), we have 
\beaa
M_t = \Eof{M_t} + \int_0^t \E_s\left[D^B_s M_T\right] dB_s.
\eeaa
We calculate $\E_s\left[D^B_s M_T\right]$ as follows. Notice that for the fSABR and for $s<t$
$$
D^B_s Y_t=\nu Y_t K(t,s),
$$
where $K(t,s)$ is the kernel defined in \eqref{eqn:molchon-golosov}. Hence, 
\beaa
&& D^B_s M_T = \int_0^T D^B_s(Y_r^2) dr \\
&=& \int_s^T D_s^B\left( Y_0^2 e^{2\nu B_r^H}\right) dr \quad \left(\because D_s^B Y_r = 0 \; \forall r < s \right) \\
&=& 2 \nu \int_s^T Y_r^2 D_s^B(B_r^H) dr \\
&=& 2 \nu \int_s^T Y_r^2 K(r,s) dr.
\eeaa
Thus, 
\[
\E_s\left[D^B_s M_T\right] = 2 \nu \int_s^T \E_s\left[Y_r^2\right] K(r,s) dr.
\]
%which, in the earlier notation, $2\nu \Eof{Y_r^2|\cF_s^B} K(r,s) = a(r,s)$ for $r > s$. 
Moreover, we have
\beaa
&& d\langle M \rangle_t = 4\nu^2 \left(\int_t^T \Et{Y_r^2} K(r,t) dr\right)^2 dt \\
&& d\langle X, M \rangle_t = 2\nu \rho \left(Y_t \int_t^T \Et{Y_r^2} K(r,t) dr\right) dt.
\eeaa
Finally, since conditioned on $\cF_t$ the fractional Brownian motion $B_r^H$ is normally distributed with mean and variance given by $m(r|t)$ and $v(r|t)$ in Lemma \ref{lma:cond-fbm}, we have
\[
\Et{Y_r^2} = Y_0^2 \Et{e^{2\nu B_r^H}} = Y_0^2 e^{2\nu m(r|t) + 2\nu^2 v(r|t)}.
\]
%In particular, when $t = 0$, we have
%\beaa
%&& m(r|0) = 0, \\
%&& v(r|t) = r^{2H}.
%\eeaa
%Thus, 
%\[
%\Eof{Y_r^2|\cF_0^B} = Y_0^2 \Eof{\left. e^{2\nu B_r^H} \right|\cF_0^B} = Y_0^2 e^{2\nu^2 r^{2H}}.
%\]

%%%%%%%%%%%%%%%%%%%%%%%%%%%%%%%%%%%%%%%%%%%%%%%%%%%%%%%%%%%%%%%%%%%%%%%%%%%%%%%%%%%%%%%%
%
%  Bibliography
%
%%%%%%%%%%%%%%%%%%%%%%%%%%%%%%%%%%%%%%%%%%%%%%%%%%%%%%%%%%%%%%%%%%%%%%%%%%%%%%%%%%%%%%%%

\end{document}